\documentclass[a4paper]{article}
\usepackage{RRA4}
\usepackage{hyperref}

\usepackage{bbm}
\usepackage{amsmath,amssymb,amsthm}
\usepackage[all,ps,dvips]{xy}
\usepackage{mathrsfs}

\def\C {\ensuremath{\mathbb{C}}}
\def\CC {\ensuremath{\mathsf{C}}}
\def\BB {\ensuremath{\mathsf{H}}}

\def\Q {\ensuremath{\mathbb{Q}}}

\def\R {\ensuremath{\mathbb{R}}}

\def\RR {\ensuremath{\mathsf{R}}}
\def\CC {\ensuremath{\mathsf{C}}}

\def\AS {\ensuremath{\mathbf{H}}}
\def\BS {\ensuremath{\mathbf{H}'}}

\def\mM{\ensuremath{\mathbf{M}}}
\def\G {\ensuremath{\mathrm{GL}}}
\def\GL{\ensuremath{{\rm GL}}}

\def\crit{\ensuremath{{\rm crit}}}
\def\reg{\ensuremath{{\rm reg}}}
\def\sing{\ensuremath{{\rm sing}}}
\def\grad{\ensuremath{{\rm grad}}}

\def\jac{\ensuremath{{\rm jac}}}

\def\y {\ensuremath{\mathbf{y}}}
\def\e {\ensuremath{\mathbf{e}}}
\def\a {\ensuremath{\mathbf{a}}}

\def\z {\ensuremath{\mathbf{z}}}
\def\x {\ensuremath{\mathbf{x}}}
\def\w {\ensuremath{\mathbf{w}}}
\def\f {\ensuremath{\mathbf{F}}}
\def\g {\ensuremath{\mathbf{g}}}
\def\h {\ensuremath{\mathbf{h}}}
\def\r {\ensuremath{\mathbf{r}}}

\def\F {\ensuremath{\mathbf{F}}}
\def\E {\ensuremath{\mathbf{E}}}
\def\X {\ensuremath{\mathbf{X}}}

\def\v {\ensuremath{\mathbf{v}}}

\newtheorem{theorem}{Theorem}
\newtheorem{corollary}[theorem]{Corollary}

\newtheorem{lemma}[theorem]{Lemma}

\RRdate{F\'evrier 2009}
\RRNo{6832}
\RRauthor{
Mohab Safey El Din\thanks{Universit\'e Pierre et Marie Curie, {\tt Mohab.Safey@lip6.fr}}%
  \and
Eric Schost\thanks{University of Western Ontario, {\tt eschost@uwo.ca}}%
}
\authorhead{Safey El Din \& Schost}
\RRtitle{Algorithme pas de b\'eb\'e/pas de g\'eant pour le calcul de cartes routi\`eres d'ensembles alg\'ebriques r\'eels compacts}
\RRetitle{A baby steps/giant steps Monte Carlo algorithm for computing roadmaps in smooth compact real hypersurfaces}
\titlehead{Roadmaps in compact real algebraic sets}
\RRresume{On consid\`ere le probl\`eme du calcul de cartes routi\`eres
  dans des ensembles alg\'ebriques r\'eels. Ce probl\`eme est
  introduit par Canny pour r\'epondre \`a des questions de connexit\'e
  et r\'esoudre des probl\`emes de planification de trajectoires.
  \'Etant donn\'ees $s$ \'equations polynomiales \`a coefficients
  rationnels, de degr\'e $D$ en $n$ variables, l'algorithme
  Monte-Carlo de Canny a une complexit\'e born\'ee par $s^n\log(s)
  D^{O(n^2)}$ op\'erations dans $\mathbb{Q}$; sa version
  d\'eterministe a une complexit\'e born\'ee par $s^n \log(s)
  D^{O(n^4)}$.  L'am\'elioration suivante est d\^ue \`a Basu, Pollack
  and Roy, dont l'algorithme d\'eterministe a u co\^ut $s^{d+1}
  D^{O(n^2)}$ pour le probl\`eme plus g\'en\'eral du calcul de cartes
  routi\`eres d'ensembles semi-alg\'ebriques ($d \le n$ est la
  dimension d'un objet alg\'ebrique associ\'e au semi-alg\'ebrique
  \'etudi\'e).

  On donne ici un algorithme Monte-Carlo de complexit\'e born\'ee par
  $(nD)^{O(n^{1.5})}$ pour le calcul de cartes routi\`eres dans une
  hypersurface r\'eelle compact $V\subset \R^n$ de degr\'e $D$ ayant
  un nombre fini de points singuliers. Sous ces hypoth\`eses, aucun
  algorithme pr\'ec\'edent n'a de co\^ut meilleur que $D^{O(n^2)}$.  }

\RRabstract{We consider the problem of constructing
  roadmaps of real algebraic sets. The problem was introduced by Canny
  to answer connectivity questions and solve motion planning
  problems. Given $s$ polynomial equations with rational coefficients,
  of degree $D$ in $n$ variables, Canny's algorithm has a Monte Carlo
  cost of $s^n\log(s) D^{O(n^2)}$ operations in $\mathbb{Q}$; a
  deterministic version runs in time $s^n \log(s) D^{O(n^4)}$.  The
  next improvement was due to Basu, Pollack and Roy, with an algorithm
  of deterministic cost $s^{d+1} D^{O(n^2)}$ for the more general
  problem of computing roadmaps of semi-algebraic sets ($d \le n$ is
  the dimension of an associated object).

We give a Monte Carlo algorithm of complexity $(nD)^{O(n^{1.5})}$ for
the problem of computing a roadmap of a compact hypersurface $V$ of
degree $D$ in $n$ variables; we also have to assume that $V$ has a
finite number of singular points. Even under these extra assumptions,
no previous algorithm featured a cost better than $D^{O(n^2)}$.
}
\RRmotcle{solutions r\'eelles de syst\`emes polynomiaux, connexit\'e, planification de trajectoires}
\RRkeyword{real solutions of polynomial systems, connectivity decision, robot motion planning}
\RRprojet{SALSA}
\RRtheme{\THSym} 
\URRocq 
\begin{document}
\makeRR   
\section{Introduction}
\paragraph{Motivation.}
Deciding connectivity properties in a semi-algebraic set $S$ is an
important problem that appears in many fields, such as motion
planning~\cite{Sharir}. This general problem is reduced to
computations in dimension 1, {\em via} the computation of a
semi-algebraic curve $\mathscr{R}$, that we call a {\em roadmap}. This
curve should have a non-empty and connected intersection with each
connected component of $S$: then, connecting two points in $S$ is done
by connecting these points to $\mathscr{R}$. Also, counting the
connected components of $S$ is reduced to counting those of
$\mathscr{R}$. Hence, a roadmap is used as the skeleton of
connectivity decision routines for semi-algebraic sets. In addition to
its direct interest, the computation of roadmaps is also used in more
general algorithms allowing us to obtain semi-algebraic descriptions
of the connected components of semi-algebraic
sets~\cite[Ch.15-16]{BaPoRo06}. Thus, improvements on the complexity
of computing roadmaps impact the complexity of many fundamental
procedures of effective real algebraic geometry.


\paragraph{Prior results.}
The notion of a roadmap was introduced by Canny in \cite{CannyThese,
  Canny}; the resulting algorithm constructs a roadmap of a
semi-algebraic set $S \subset \R^n$ defined by $k$ equations and $s$ inequalities
of degree bounded by $D$, but does not construct a path linking points
of $S$. Its complexity is $s^n\log(s)D^{O(n^4)}$ arithmetic
operations, and a Monte Carlo version of it runs in time
$s^n\log(s)D^{O(n^2)}$ (to estimate running times, we always use
arithmetic operations). Several subsequent works \cite{HRSRoadmap,
  GRRoadmap} gave algorithms of cost $(sD)^{n^{O(1)}}$; they culminate
with the algorithm of Basu, Pollack and Roy~\cite{BaPoRo96,BPRRoadmap} of cost
$s^{d+1}D^{O(n^2)}$, where $d$ is the dimension of the algebraic set
defined by the $k$ equations. These algorithms reduce the general
problem to the construction of a roadmap in a bounded and smooth
hypersurface defined by a polynomial $f$ of degree $D$; the
coefficient of $f$ lie in a field $\bf Q$ that contains several
infinitesimal quantities (it is a purely transcendental extension of
$\Q$).

Under the smoothness and compactness assumptions, and even in the
simpler case of a polynomial $f$ with coefficients in $\Q$, none of
the previous algorithms features a cost lower than $D^{O(n^2)}$ and
none of them returns a roadmap of degree lower than $D^{O(n^2)}$. In
this paper, we give the first known estimates of the form
$(nD)^{O(n^{1.5})}$ for this particular problem, in terms of output
degree and running time.

All these previous works, and ours also, make use of computations of
critical loci of projections and rely on geometric connectivity
results for correctness. Before recalling the basics we need about
algebraic sets and critical loci, we give precise definitions of
roadmaps and state our main result.

\paragraph{Definitions and main result.}
The original definition (found in~\cite{BaPoRo06}) is as follows.  Let
$S$ be a semi-algebraic set. A {\em roadmap} for $S$ (in the sense
of~\cite{BaPoRo06}) is a semi-algebraic set $\mathscr{R}$ of dimension
at most $1$ contained in $S$ which satisfies the following conditions:
\begin{itemize}
\item[${\rm RM}_1$] Each connected component of $S$ has a non-empty
  and connected intersection with $\mathscr{R}$.
\item[${\rm RM}_2$] For $x\in \R$, each connected component of
  $S_x$ intersect $\mathscr{R}$, where $S_x$ is the set of points of
  the form $(x,x_2,\dots,x_n)$ in $S$.
\end{itemize}
\noindent We modify this definition (in particular by discarding ${\rm
  RM}_2$), for the following reasons. First, it is
coordinate-dependent: if $\mathscr{R}$ is a roadmap of $S$, it is not
necessarily true that $\phi(\mathscr{R})$ is a roadmap of $\phi(S)$,
for a linear change of coordinates $\phi$. Besides, one interest of
${\rm RM}_2$ is to make it possible to connect two points in $S$ by
adding additional curves to $\mathscr{R}$: condition ${\rm RM}_2$ is
well-adjusted to the connecting procedure given in~\cite{BaPoRo06}, 
which we do not use here.

\noindent 
Hence, we propose a modification in the definition of roadmaps. We do
not deal with semi-algebraic sets, but only with sets of the form $V
\cap \R^n$, where $V \subset \C^n$ is an algebraic set. Our
definition, like the previous one, allows us to count connected
components and to construct paths between points in $V \cap \R^n$. We
generalize the definition to higher-dimensional ``roadmaps'', since
our algorithm computes such objects. Thus, we say that an algebraic
set $\mathscr{R} \subset \C^n$ is an $i$-roadmap of $V$ if:
\begin{itemize}
\item[${\rm RM}'_1$] Each connected component of $V\cap \R^n$ has a
  non-empty and connected intersection with $\mathscr{R}\cap \R^n$.
\item[${\rm RM}'_2$] The set $\mathscr{R}$ is contained in $V$.
\item[${\rm RM}'_3$] The set $\mathscr{R}$ has dimension $i$.
\end{itemize}
If $\dim(\mathscr{R})=1$, we simply say that $\mathscr{R}$ is a
roadmap of $V$. Finally, it will be useful to add a finite set of
control points $\mathscr{P}$ to our input, e.g. to test if the points
of $\mathscr{P}$ are connected on $V\cap \R^n$. Then, $\mathscr{R}$ is
a $i$-roadmap of $(V,\mathscr{P})$ if we also have:
\begin{itemize}
\item[${\rm RM}'_4$] The set $\mathscr{R}$ contains $\mathscr{P}$.
\end{itemize}
\noindent Hereafter, given a finite set $\mathscr{P}$, we write its
cardinality $\delta_P$ (if $\mathscr{P}=\emptyset$, we take
$\delta_P=1$).
\begin{theorem}
  Given $f \in \Q[X_1,\dots,X_n]$ such that $V(f)\cap \R^n$ is compact
  and has a finite number of singular points, and given a subset
  $\mathscr{P}$ of $V(f)$ of cardinality $\delta_P$, one can compute a
  roadmap of $(V(f),\mathscr{P})$ of degree
  $\delta_{P}(nD)^{O(n^{1.5})}$ in Monte Carlo time
  ${\delta_{P}}^{O(1)}(nD)^{O(n^{1.5})}$.
\end{theorem}
The probabilistic aspects of our algorithm are twofold: first, we
choose random changes of variables to ensure nice geometric
properties. Second, we need to solve systems of polynomial equations;
for our purpose, the algorithm with the best adapted
cost~\cite{Lecerf00} is probabilistic as well. Remark that we can also
deterministically compute a roadmap of $(V(f),\mathscr{P})$ of degree
$\delta_{P}(nD)^{O(n^{1.5})}$: exhaustive searches in a large enough
sample set enable us to deterministically find suitable changes of
variables; then, deterministic polynomial system solving algorithms
replace the use of~\cite{Lecerf00}.

We expect in further work to apply our techniques to the case where
the input polynomial has coefficients in a field that contains
infinitesimal quantities: similar generalizations, based on the
Transfer Principle, are in~\cite[Ch.~12]{BaPoRo06}. We hope to obtain
general roadmap algorithms for semi-algebraic sets of cost
$s^{O(d)}(nD)^{{O}(n^{1.5})}$ (using the notation of the previous
paragraphs).

\paragraph{Algebraic sets.}
To describe our contribution, we need a few definitions.  We define
most of the notation needed below; for standard notions not recalled
here, see~\cite{ZaSa58,Mumford76,Shafarevich77,Eisenbud95}. An {\em
  algebraic set} $V \subset \C^n$ is the set of common zeros of some
polynomial equations $f_1,\dots,f_s$ in variables $X_1,\dots,X_n$; we
write $V=V(f_1,\dots,f_s)$.  The dimension of $V$ is the Krull
dimension of $\C[X_1,\dots,X_n]/I$, where $I$ is the ideal $\langle
f_1,\dots,f_s \rangle$ in $\C[X_1,\dots,X_n]$. The set $V$ can be
uniquely decomposed into {\em irreducible} components, which are
algebraic sets as well; when they all have the same dimension, we say
that $V$ is {\em equidimensional}. The {\em degree} of an irreducible
algebraic set $V \subset \C^n$ is the maximum number of intersection points
between $V$ and a linear space of dimension $n-\dim(V)$; the degree of
an arbitrary algebraic set is the sum of the degrees of its
irreducible components.

The tangent space to $V$ at $\x \in V$ is the vector space $T_\x V$
defined by the equations $\grad(f,\x)\cdot\v =0$, for all polynomials
$f$ that vanish on $V$. When $V$ is equidimensional, the {\em regular
  points} on $V$ as those points $\x$ where $\dim(T_\x V)=\dim(V)$;
more generally, the regular points are those where the local ring of
$V$ at $\x$ is regular of dimension $d$. The {\em singular points} are
all other points. The set of regular (resp. singular) points is
denoted by $\reg(V)$ (resp. $\sing(V)$). The set $\sing(V)$ is an
algebraic subset of $V$, of smaller dimension than $V$.

\paragraph{Polar varieties.}
Canny's algorithm is the best known approach to computing roadmaps. 
Given an algebraic set $V$, it
proceeds by computing some critical curves on $V$, and studying some
distinguished points on these curves. One of our  contributions is
the use of higher-dimensional critical loci, called {\em polar
  varieties}, that were introduced by Todd~\cite{Todd37} and studied
from the algorithmic point of view
in~\cite{BaGiHeMb97,BaGiHeMb01}. For positive integers $i \le n$, we
denote by $\Pi_i$ the projection
$$\begin{array}{cccc}
\Pi_i: & \C^n & \to & \C^i \\
 & \x =(x_1,\dots,x_n) & \mapsto & (x_1,\dots,x_i).
\end{array}$$
Then, the polar variety $w_i$ is the set of critical points of $\Pi_i$
on $\reg(V)$, that is, the set of all points $\x \in \reg(V)$ such
that $\Pi_i(T_\x V) \ne \C^i$. The set $w_i$ may not be an
algebraic set if $V$ has singular points; in this case, however, the set $W_i = w_i \cup \sing(V)$
is algebraic. By abuse of notation, we still call it a polar 
variety and we write $W_i=\crit(\Pi_i,V)$. Its expected dimension is $i-1$.

If $V$ is given as $V(f_1,\dots,f_p)$, is equidimensional of dimension
$d=n-p$, and if the ideal $\langle f_1,\dots,f_p \rangle$ is radical,
then $W_i$ is the zero-set of $(f_1,\dots,f_p)$ and of all minors of
size $p$ taken from the jacobian matrix $\jac(\f,[X_{i+1},\dots,X_n])$ of $\f$ in $X_{i+1},\dots,X_n$.

\paragraph{Using polar varieties.}
Given $f$ of degree $D$ and $V=V(f)$, assuming $V(f)\cap\R^n$ is
smooth and compact, Canny's algorithm computes the critical curve
$W_2=\crit(\Pi_2,V)$. The compactness assumption ensures that $W_2$
intersects each connected component of $V\cap \R^n$, but not that
these intersections are connected. The solution consists in choosing a
suitable family $\mathscr{E}=\{x_1, \ldots, x_N\} \subset\R$ so that
the union of $W_2$ and $\mathscr{C}'=V \cap \Pi_1^{-1}(\mathscr{E})$
is an $(n-2)$-roadmap of $V$.

To realize this, Canny's algorithm uses the following connectivity
result: defining the (expectedly finitely many) points $\mathscr{C}=
W_1\cup \crit(\Pi_1, W_2)$, and taking their projection
$\mathscr{E}=\Pi_1(\mathscr{C})$ in the construction above gives an
$(n-2)$-roadmap of $V$ of degree $D^{O(n)}$. Then, the algorithm
recursively constructs a roadmap in $\Pi_1^{-1}(\mathscr{E})\cap V$
following the same process; this is geometrically equivalent to a
recursive call with input $f(x, X_2, \ldots, X_n)$ for all $x\in
\mathscr{E}$.  At each recursive call, the number of control points we
compute is multiplied by $D^{O(n)}$, but the dimension of the input
decreases by $1$ only. Thus, the depth of the recursion is $n$ and the
roadmap we get has degree $D^{O(n^2)}$.

Our algorithm relies on a new connectivity result that generalizes the
one described above. We want to avoid the degree growth by performing
recursive calls on inputs whose dimension has decreased by $i \gg
1$. To this end, instead of considering, as Canny did, the polar curve
$W_2$ associated to a projection on a plane, we use polar varieties
$W_i$ of higher dimension. As above, we have to consider suitable
fibers $V \cap \Pi_{i-1}^{-1}(\x)$ to repair the defaults of
connectivity of $W_i$. To achieve this, we use the following new
result (Theorem~\ref{theo:big} below): define $\mathscr{C}=W_1 \cup
\crit(\Pi_1, W_i)$ and
$\mathscr{C}'=V\cap\Pi_{i-1}^{-1}(\Pi_{i-1}(\mathscr{C}))$; under some
crucial (but technical) assumptions, $W_i\cup\mathscr{C}'$ is a
$\max(i-1,n-i)$-roadmap of $V$. This leads to a more complex recursive
algorithm; the optimal cut-off we could obtain that ensured all
necessary assumptions has $i \simeq \sqrt{n}$.

\paragraph{Outline of the paper; basic notation.} Our algorithm is
described in the next section. The final two sections sketch the
proofs of two key points: the connectivity result mentioned above, and
the fact that generic changes of variables suffice to ensure the
assumptions needed by this connectivity result.

If $X$ is a subset of either $\C^n$ or $\R^n$, and if $A$ is a subset
of $\R$, we write $X_A\ =\ X \cap \Pi_1^{-1}(A)\cap \R^n.$ For $x$ in
$\R$, we use the particular cases $X_{<x} = X_{]-\infty,x)}$, $X_{x} =
X_{\{x\}}$, $X_{\le x} = X_{]-\infty,x]}.$ Hereafter, a property is
called {\em generic} if it holds in a Zariski-open dense subset of the
corresponding parameter space.


\section{Algorithm}

Even though we are interested in roadmaps for hypersurfaces, the
recursive structure of the algorithm requires that we consider systems
of the form $\f=(f_1,\dots,f_p)$ in $\Q[X_1,\dots,X_n]$. After stating
our connectivity result, we give a modification of Canny's algorithm
for such systems, then use it as a subroutine for our main algorithm.


\subsection{Main connectivity result and sketch of the algorithm}\label{ssec:main}

We say that the system $\f$ satisfies assumption $\AS$ if
\begin{itemize}
\item[$(a)$] the ideal $\langle f_1,\dots,f_p \rangle$ is radical;
\item[$(b)$] $V=V(f_1,\dots,f_p)$ is equidimensional of dimension $d=n-p$;
\item[$(c)$] $\sing(V)$ is finite;
\item[$(d)$] $V \cap \R^n$ is bounded.
\end{itemize}
These conditions are independent of the choice of coordinates. Next,
we fix $i$ in $\{2,\dots,d-1\}$ and we say that $\f$ satisfies
condition $\BS$ if the following holds:
\begin{itemize}
\item[$(a)$] $\dim(V)=d$ and the extension $\C[X_1,\dots,X_d] \to
  \C[X_1,\dots,X_n]/\langle f_1,\dots,f_p\rangle$ is integral ({\it
    i.e.} $V$ is in Noether position for $\Pi_d$);
\item[$(b)$] $W_i$ is in Noether position for $\Pi_{i-1}$ (same definition as above);
\item[$(c)$] $W_1$ is finite;
\item[$(d)$] $\crit(\Pi_1,W_i)$ is finite.
\end{itemize}
We will see that these new assumptions can be ensured by a generic change
of variables for some values of $p$ and $i$ (but not all). Finally, we
consider a finite subset of points $\mathscr{P}$ in $V$, and we define
\begin{itemize}
\item $\mathscr{C} = W_1\,\cup \, \crit(\Pi_1, W_i)\, \cup\, \mathscr{P}$, which is finite under 
  $\AS$ and $\BS$;
\item $\mathscr{C}'=V\,\cap\,\Pi_{i-1}^{-1}(\Pi_{i-1}(\mathscr{C}))$, so that
 $\x \in V$ is in $\mathscr{C}'$ if and only if $\Pi_{i-1}(\x)$ is in $\Pi_{i-1}(\mathscr{C})$.
\end{itemize}
The following theorem is proved in the next section; it is the key to our algorithms.
\begin{theorem}\label{theo:big}
  Let $d' = \max(i-1, d-i+1)$. Under assumptions $\AS$ and $\BS$, the
  following holds:
  \begin{enumerate}
  \item $\mathscr{C}'\cup W_i$ is a $d'$-roadmap of
    $(V,\mathscr{P})$;
  \item $\mathscr{C}'\cap W_i$ is finite;
  \item for all $\x\in \C^{i-1}$, the system
    $(f_1,\dots,f_p,X_1-x_1,\dots,X_{i-1}-x_{i-1})$ satisfies
    assumption $\AS$.
  \end{enumerate}
\end{theorem}
\noindent The idea of the algorithm is to compute $W_i$, the finite
sets $\mathscr{C}$ and $\mathscr{C}'\cap W_i$ and to recursively
compute roadmaps of $\mathscr{C}'$ and $W_i$, if their dimension is
too high. For $\mathscr{C}'$, this will be possible by point~3 of the
theorem, but for $W_i$ this will be more delicate, since we may not be
able to enforce $\AS$; this will restrict our choices for
$i$ and dictate the structure of the algorithm. The correctness of this recursive process follows from the
following lemma.
\begin{lemma}\label{lemma:glue}
  With notation as above, if $\mathscr{R}_1$ and $\mathscr{R}_2$ are
  roadmaps of respectively $(W_i, (\mathscr{C}' \cap W_i) \cup
  \mathscr{P})$ and $(\mathscr{C}', (\mathscr{C}' \cap W_i)\cup
  \mathscr{P})$, then $\mathscr{R}_1 \cup \mathscr{R}_2$ is a
  roadmap of $(V,\mathscr{P})$.
\end{lemma}


\subsection{Preliminaries to the algorithms}\label{ssec:preli}

\paragraph{Data representation.}
The outputs of our algorithms are sets of {\em rational
  parametrizations} of algebraic curves: if $\mathcal{C} \subset \C^n$
is an algebraic curve defined over $\Q$, such a parametrization
consists in polynomials $Q=(q,q_0,\dots,q_n)$ in $\Q[U,T]$ and two linear
forms $\tau=\tau_1 X_1 +\cdots +\tau_n X_n$, $\eta=\eta_1 X_1 +\cdots
+\eta_n X_n$ with coefficients in $\Q$, such that $\mathcal{C}$ is the
Zariski closure of the set defined by
$$q(\eta,\tau) = 0, \qquad X_i = \frac{q_i(\eta, \tau)}{q_0(\eta, \tau)} \ \ (1 \le i \le n),
\qquad q_0(\eta,\tau) \ne 0.$$ The degree of the curve $\mathcal{C}$
is written $\delta_Q$; then, all polynomials in $Q$ can be taken of
degree $\delta_Q^{O(1)}$. By a slight abuse of language, we will say
that a family of 1-dimensional parametrizations is a roadmap of a set
$V$ if the union of the curves they define is.

Finally, internally to the algorithm, we use a similar notion for
0-dimensional ({\it i.e.}  finite) sets of points; then, all
polynomials involved are univariate, and a single linear form is
needed~\cite{GiHe93,Rouillier99}. In this case, we write
$\delta_Q$ for the number of points described by $Q$. If $Q$
represents a set of points in $\C^e$ in variables $X_1,\dots,X_e$, we
write $Q(X_1,\dots,X_e)$.

\medskip\noindent{\bf Quantities carried through recursive calls.}  To
accommodate the recursive nature of the algorithm, we take as input a
pair $[\f,Q]$, where $\f$ is as before and $Q(X_1,\dots,X_e)$ is a
$0$-dimensional parametrization. We are interested in roadmaps of
$V(\f,Q)$; this means that we restrict $X_1,\dots,X_e$ to a finite
number of possible values, that are solutions of $Q$

In this new context, we define analogues of $\AS$ and
$\BS$. Assumption $\AS$ remains unchanged for $[\f,Q]$, up to
replacing $V(\f)$ by $V(\f,Q)$ and $n-p$ by $n-p-e$. To state $\BS$,
for $\x=(x_1,\dots,x_e)$ in $V(Q)$, we define
$\f_\x=\f(x_1,\dots,x_e,Y_1,\dots,Y_{n-e}) \in \C[Y_1,\dots,Y_{n-e}]$,
for some new variables $Y_1,\dots,Y_{n-e}$.  Then we say that $[\f,Q]$
satisfies $\BS$ if for all $\x$ in $V(Q)$, $\f_\x$ satisfies $\BS$.

\medskip\noindent{\bf Subroutines.}  We use a function {\sf Solve} for
solving 0- and 1-dimensional polynomial systems; the result is a
rational parametrization of the solutions. If the input has $s$
equations of degree at most $D$ (with $D\ge 1$), the algorithm
of~\cite{Lecerf00} performs this task in time $sD^{O(n)}$.  The
function {\sf Union} (resp. {\sf Projection}) computes a
parametrization of the union (resp. a projection) of two (resp. one)
0-dimensional sets given by parametrizations; on inputs of degree at
most $\delta$, this takes time $\delta^{O(1)}$.  Finally, we need
algorithms for computing critical points, on two slightly different
kinds of inputs:
\begin{itemize}
\item Given a parametrization $R$ of a curve $\mathcal{C}$ in $\C^n$,
  ${\sf CriticalPointsCurve}(R,X_j)$ computes
  $\crit(\Pi_{X_j},\mathcal{C})$, where $\Pi_{X_j}$ is the projection
  on the $X_j$-axis. Due to the nice shape of our parametrizations,
  this can be done in time $\delta_R^{O(1)}$.
\item Given a system $[\f,Q]$ that satisfies $\AS$, ${\sf
    CriticalPoints}(\f,X_j)$ computes $\crit(\Pi_{X_j},V(\f,Q))$. This
  time, assumption $\AS$ makes it possible to use directly the
  Jacobian matrix of $\f$ to perform this operation; this can be done
  in time $\delta_Q^{O(1)}(nD)^{O(n)}$.
\end{itemize}


\subsection{Canny's algorithm revisited}\label{ssec:canny}

We start with an algorithm close to Canny's. As opposed to Canny, we
do not work with a single equation but with a system
$\f=f_1,\dots,f_p$ that satisfies $\AS$; as Canny, we take $i=2$ in
the recursion. Indeed, given such a system, we will see that it is
possible to ensure assumption $\BS$ through a generic change of
variables for $i=2$, but not for $i>2$.  As said above, we take a
0-dimensional parametrization $Q(X_1,\dots,X_e)$ as input as well;
then, our change of variables $\varphi$ will leave $X_1,\dots,X_e$
fixed and we denote by $\G(n,e)$ the subset of $\GL_n(\Q)$ satisfying
this constraint.  Our last input are the control points $\mathscr{P}$,
given in the form of a 0-dimensional parametrization $P$.

\begin{lemma}\label{prop:3}
  Suppose that $[\f,Q]$ satisfies $\AS$. After a generic change of
  variables in $\G(n,e)$, the system $[\f,Q]$ satisfies $\AS$ and
  $\BS$ for $i=2$.
\end{lemma}
\noindent {\sf CannyRoadmap}$(\f,Q,P)$.
\begin{enumerate}
\item[0.] If $n-p-e=1$, return ${\sf Solve}([\f,Q])$
\item Apply a random change of variables $\varphi \in \G(n,e)$
\item Let $\Delta =[ \text{$p$-minors~of~} \jac(\f,[X_{e+2},\dots,X_n])]$ and 
          $\Delta' =[ \text{$p$-minors~of~} \jac(\f,[X_{e+3},\dots,X_n])]$
\item Let $R= \text{{\sf Solve}}([\f,\Delta,Q])$ and $R'= \text{{\sf Solve}}([\f,\Delta',Q])$
\item Let $S=\text{{\sf CriticalPointsCurve$(R',X_{e+1})$}}$
\item Let $Q'={\sf Projection}({\sf Union}([S, R, P]),[X_1,\dots,X_{e+1}])$
\item Let $P'={\sf Union}({\sf Solve}([\f,\Delta',Q']), P)$
\item Let $R''={\sf CannyRoadmap}(\f, Q', P')$ \hfill ($e$ increases by 1)
\item Undo the change of variables $\varphi$ and return $(R',R'')$
\end{enumerate}
\noindent To understand the algorithm, it is easier to consider that
$Q$ is empty and thus $e=0$. Under $\AS$ and $\BS$, the sets $R$ and
$R'$ respectively describe $W_1$ and $W_2$, and $S$ describes
$\crit(\Pi_1,W_2)$. Then, $Q'$ encodes the set $\Pi_{1}(\mathscr{C})$
of Subsection~\ref{ssec:main}, and $P'$ is the new set of control
points $(\mathscr{C}'\cap W_1) \cup \mathscr{P}$. The algebraic set
$V(\f,Q')$ equals $\mathscr{C}'$, to which we recursively apply {\sf
  CannyRoadmap}. Since $R'$ describes $W_2$ and $W_2$ is a curve, there
is no need to process it further, and we append it to the output.

\begin{lemma}\label{complexiteCanny}
  {\sf CannyRoadmap} computes a roadmap of $(V(\f,Q), \mathscr{P})$ of
  degree $(\delta_Q+ \delta_{P}) (nD)^{O(n(n-p-e))}$ in Monte Carlo
  time $(\delta_Q+ \delta_{P})^{O(1)} (nD)^{O(n(n-p-e))}$.
\end{lemma}
\noindent Once $\AS$ and $\BS$ hold, correctness follows from
Theorem~\ref{theo:big}; the domain where we pick $\varphi$ is
discussed in appendix p.~\pageref{sec:probasp}. To estimate runtime,
one first notes that the cost of steps $0-6$ is $(\delta_{Q} +
\delta_{P})^{O(1)}(nD)^{O(n)}$, and that we have $\delta_{Q'} +
\delta_{P'} \le (\delta_{Q} + \delta_{P}) (nD)^{O(n)}$. Since the
depth of the recursion is $n-p-e$, this proves our claims. Remark that
for $e=0$ and $p=1$, we recover Canny's result.


\subsection{Main algorithm}\label{ssec:us}

We finally give our roadmap algorithm for a hypersurface $V(f)$, where
$f$ satisfies assumption $\AS$. Here, we can ensure assumption $\BS$
in generic coordinates for many more choices of $i$. Using our
modified version of Canny's algorithm, we obtain a baby steps/giant
steps strategy by choosing $i \simeq \sqrt{n}$. As before, we also
take a 0-dimensional parametrization $Q(X_1,\dots,X_e)$ as input, and
the control points $\mathscr{P}$ by means of a 0-dimensional
parametrization $P$.

\medskip\noindent {\sf Roadmap}$(f,Q,P)$.
\begin{enumerate}
\item[0.] If $n-p-e \le \sqrt{n}$, return {\sf CannyRoadmap}$(f,Q,P)$
\item Let $i=\lfloor \sqrt{n} \rfloor$
\item Apply a random change of variables $\varphi \in \G(n,e)$.
\item Let $\Delta =[ \partial f/ \partial X_i \ | \ i \in [{e+2},\dots,n]]$, 
          $\Delta' =[ \partial f/ \partial X_i \ | \ i \in [{e+i+1},\dots,n]]$ and
          $\f=(f,\Delta')$
\item Let $R= \text{{\sf Solve}}([f,\Delta,Q])$.
\item Let $S=\text{{\sf CriticalPoints$([\f,Q],X_{e+1})$}}$
\item Let $Q'={\sf Projection}({\sf Union}([S, R, P]),[X_1,\dots,X_{e+i-1}])$
\item Let $P'={\sf Union}({\sf Solve}([\f,Q']), P)$
\item Let $R''={\sf CannyRoadmap}(\f, Q, P')$
\item Let $R'''={\sf Roadmap}(f, Q', P')$ \hfill ($e$ increases by $\lfloor \sqrt{n} \rfloor$)
\item Undo the change of variables $\varphi$ and return $(R'',R''')$
\end{enumerate}
\begin{lemma}\label{prop:4} {\em (using the notation of the algorithm)}
  Suppose that $[f,Q]$ satisfies $\AS$. For $i \le n-p-e-1$, after a
  generic change of variables in $\G(n,e)$, $[f,Q]$ satisfies $\AS$
  and $\BS$ and $[\f,Q]$ satisfies $\AS$.
\end{lemma}
\noindent 
As before, we explain the computations with $Q$ empty, so $e=0$. Under
$\AS$ and $\BS$, $R$ describes $W_1$ and $S$ describes
$\crit(\Pi_1,W_i)$; we do not actually compute a parametrization for
$W_i$, since the equations $\f$ are well adapted for this
computation. Then, $Q'$ encodes the set $\Pi_{i-1}(\mathscr{C})$ of
Subsection~\ref{ssec:main}, and $P'$ is the new set of control points
$(\mathscr{C}'\cap W_i) \cup \mathscr{P}$. The equations $(f,Q')$
describe $\mathscr{C}'$, to which we recursively apply {\sf
  Roadmap}. The equations $\F$ describe $W_i$, to which we apply the
algorithm {\sf CannyRoadmap} of the last section (this is valid, since
$\F$ satisfies $\AS$).
\begin{lemma}\label{complexite}
  {\sf Roadmap} computes a roadmap of $(V(f,Q), \mathscr{P})$ of
  degree $(\delta_Q+ \delta_{P}) (nD)^{O(n^{1.5})}$ in Monte Carlo
  time $(\delta_Q+ \delta_{P})^{O(1)} (nD)^{O(n^{1.5})}$.
\end{lemma}
\noindent As for {\sf CannyRoadmap}, correctness follows from
Theorem~\ref{theo:big}. Initially, we take $Q$ and $P$ of degrees
$\delta_{Q}$ and $\delta_{P}$. After $r$ recursive calls, we have $e
\simeq r\sqrt{n}$, the degrees of the ``local'' $Q$ and $P$ have order
$(\delta_{Q}+\delta_{P}) (nD)^{O(nr)}$, and the cost of the
computation is $(\delta_{Q}+\delta_{P})^{O(1)} (nD)^{O(nr)}$. We enter
the function {\sf CannyRoadmap} with $n-p-e\simeq \sqrt{n}$, so the
cost of this call is $(\delta_{Q}+\delta_{P})^{O(1)} (nD)^{O(nr)}
(nD)^{O(n^{1.5})}$, and the degree its output is
$(\delta_{Q}+\delta_{P}) (nD)^{O(nr)} (nD)^{O(n^{1.5})}$. Since the
depth of the recursion is $r=O(\sqrt{n})$, this gives the result
claimed in the introduction.


\section{Proof of the connectivity result}\label{sec:proofbig}

We sketch the proof of the first point of Theorem~\ref{theo:big}. We
focus on the connectivity property ${\rm RM}_1'$, which is the
hardest; the missing arguments are in appendix.

We reuse here the notation of Theorem~\ref{theo:big} and we let
$\mathscr{R}=\mathscr{C}'\cup W_i$. For $x$ in $\R$, we say that
property ${\bf P}(x)$ holds if for any connected component $C$ of
$V_{\le x}$, $C \cap \mathscr{R}$ is non empty and connected. We will
prove that for all $x$ in $\R$, ${\bf P}(x)$ holds; taking $x \ge
\max_{\y\in V\cap \R^n} \Pi_1(\y)$ gives our result. To do so, we let
$v_1 < \cdots < v_\ell$ be the projections $\Pi_1(v)$, for $v$ in
$\mathscr{C} \cap \R^n$ (recall that $\mathscr{C}$ is finite). The
proof uses two intermediate results:
\begin{itemize}
\item if ${\bf P}(v_j)$ holds, then for $x$ in $(v_j,v_{j+1})$, then
  ${\bf P}(x)$ holds;
\item for $x$ in $\R$, if ${\bf P}(x')$ holds for all $x' < x$, then
  ${\bf P}(x)$ holds.
\end{itemize}
Since for $x < \min_{\y\in V\cap \R^n} \Pi_1(\y)$, property ${\bf
  P}(x)$ vacuously holds, the combination of these two results gives
the claim above by an immediate induction.  

\medskip\noindent
As preliminaries, we consider an algebraic set $Z \subset \C^n$; for
$x\in \R$, we are interested in the properties of the connected
components of $Z_{<x}$ in the neighborhood of the hyperplane
$\Pi_1^{-1}(x)$. The following result actually holds for $Z$ in
$\CC^n$, where $\CC$ is the algebraic closure of a real closed field
$\RR$.
\begin{lemma}\label{prop:wibar}
  Let $x$ be in $\R$ and let $\gamma:A \to Z_{\le x}-Z_x\cap
  \crit(\Pi_1,Z)$ be a continuous semi-algebraic map, where $A \subset
  \R^k$ is a non-empty connected semi-algebraic set. Then there exists
  a unique connected component $B$ of $Z_{<x}$ such that $\gamma(A)
  \subset \overline{B}$.
\end{lemma}
\noindent We continue with a statement in the vein of Morse's
Lemma~A~\cite[Th.~7.5]{BaPoRo06}. The proof uses Ehresmann's fibration
theorem (which relies on the integration of vector fields), so we need
here our base fields to be~$\R$ and $\C$.
\begin{lemma}\label{prop:ret}
  Suppose that $\dim(Z)>0$ and that $Z\cap \R^n$ is compact. Let $v <
  w$ be in $\R$ such that $Z_{(v,w]} \cap \crit(\Pi_1,Z)=\emptyset$,
  and let $C$ be a connected component of $Z_{\le w}$. Then, for all
  $x$ in $[v,w]$, $C_{\le x}$ is a connected component of $Z_{\le x}$.
\end{lemma}

\medskip\noindent We can then prove our claims. We start by the
easier case: extending ${\bf P}$ from $v_j$ to $(v_j,v_{j+1})$.

\begin{lemma}
  Let $j$ be in $\{1,\dots,\ell-1\}$. If ${\bf P}(v_j)$ holds, then
  for $x$ in $(v_j,v_{j+1})$, ${\bf P}(x)$ holds.
\end{lemma}
\begin{proof}
Let $x$ be in $(v_j,v_{j+1})$ and let $C$ be a connected component of
$V_{\le x}$. We have to prove that $C \cap \mathscr{R}$ is non-empty
and connected. We first establish that $C_{\le v_j} \cap \mathscr{R}$
is non-empty and connected. Because there is no point in $W_1$ in
$V_{(v_j,x]}$, applying Lemma~\ref{prop:ret} to $V$ above the interval
$(v_j,x]$ shows that $C_{\le v_j}$ is a connected component of $V_{\le
  v_j}$. So, using property ${\bf P}(v_j)$, we see that $C_{\le v_j}
\cap \mathscr{R}$ is non-empty and connected, as needed.
  
Next, we prove that for any connected component $D$ of $C \cap W_i$,
$D_{\le v_j}$ is non-empty (and connected). Clearly, $D$ is a
connected component of ${W_i}_{\le x}$. Recall that $W_i$ is an
algebraic set of positive dimension, with $W_i \cap \R^n$ compact;
besides, $\crit(\Pi_1,W_i)$ is empty above $(v_j,x]$. Applying
Lemma~\ref{prop:ret} to $Z=W_i$, we see that $D_{\le v_j}$ is
non-empty (and connected).

To prove that $C \cap \mathscr{R}$ is connected, we prove that any
$\y$ in $C \cap \mathscr{R}$ can be connected to a point in $C_{\le
  v_j}\cap \mathscr{R}$ by a path in $C \cap \mathscr{R}$. This is
sufficient to conclude, since we have seen that $C_{\le v_j}\cap
\mathscr{R}$ is connected.  Let thus $\y$ be in $C \cap
\mathscr{R}$. If $\y$ is in $C_{\le v_j} \cap \mathscr{R}$, we are
done. If $\y$ is in $C_{(v_j,x]} \cap \mathscr{R}$, it is actually in
$C_{(v_j,x]} \cap W_i$, since $\mathscr{R}$ and $W_i$ coincide above
$(v_j,x]$. Let thus $D$ be the connected component of $C \cap W_i$
containing $\y$. By the result of the previous paragraph, there exists
a continuous path connecting $\y$ to a point $\y'$ in $D_{\le v_j}$ by
a path in $D$. Since $D$ is in $C \cap \mathscr{R}$, we are done.
\end{proof}

\begin{lemma}\label{Prop:Px2}
  Let $x$ be in $\R$ such that for all $x' < x$, ${\bf P}(x')$ holds.
  Then ${\bf P}(x)$ holds.
\end{lemma}
\begin{proof}
  Let $C$ be a connected component of $V_{\le x}$; we have to prove
  that $C\cap \mathscr{R}$ is connected. If $\dim(C)=0$, we are done,
  since $C$ is a point and $C \cap \mathscr{R}$ is connected as it is
  non-empty (one checks that $C$ is in $W_1$). Hence, we assume that
  $\dim(C) > 0$; from this, one deduces that $C_{<x}$ is not
  empty. Let then $B_1,\dots,B_r$ be the connected components of
  $C_{<x}$; Lemma~\ref{7} (in appendix) proves that for $i \le r$,
  $\overline{B_i} \cap \mathscr{R}$ is non-empty and connected.

  Since $\overline{B_1} \cap \mathscr{R}$ is non-empty and contained
  in $C\cap \mathscr{R}$, the latter is non-empty. Let thus finally
  $\y_1$ and $\y_2$ be in $C \cap \mathscr{R}$; we need to connect
  them by a path in $C \cap \mathscr{R}$.  Let $\gamma:[0,1]\to C$ be
  a continuous semi-algebraic path that connects $\y_1$ to $\y_2$, and
  let $G=\gamma^{-1}(C_x \cap W_1)$ and $H=[0,1]-G$. The connected
  components $g_1,\dots,g_N$ of $G$ are intervals and closed in
  $[0,1]$ (and may be reduced to single points); the connected
  components $h_1,\dots,h_M$ of $H$ are intervals that are open in
  $[0,1]$. Besides, these intervals are interleaved in $[0,1]$.  For
  $1 \le i \le M$, we write $\ell_i ={\rm inf}(h_i)$ and $r_i={\rm
    sup}(h_i)$; we also introduce $r_0=0$ and $\ell_{M+1}=1$.  To
  conclude the proof, we establish that:
  \begin{enumerate}
  \item for $1 \le i \le M$, $\gamma(\ell_i)$ and $\gamma(r_i)$ can be
    connected by a semi-algebraic path in $C \cap \mathscr{R}$;
  \item for $0 \le i \le M$, $\gamma(r_i)$ and $\gamma(\ell_{i+1})$ can be
    connected by a semi-algebraic path in $C \cap \mathscr{R}$.
  \end{enumerate}
  \noindent We prove the first point (the second one is easier). For
  $1 \le i \le M$, we first claim that there exists $j \le r$ such
  that $\gamma(h_i)$ is in $\overline{B_j}$. Indeed, remark that since
  $\gamma(h_i)$ avoids $C_x \cap W_1$, it actually avoids the whole
  $V_x \cap W_1$ (because $\gamma(h_i)$ is contained in $C$). It
  follows from Lemma~\ref{prop:wibar} that there exists a connected
  component $B_j$ of $V_{<x}$ such that $\gamma(h_i)\subset
  \overline{B_j}$. Since $\gamma$ is continuous, both $\gamma(\ell_i)$
  and $\gamma(r_i)$ are in $\overline{B_j}$.  On the other hand, both
  $\gamma(\ell_i)$ and $\gamma(r_i)$ are in $\mathscr{R}$. We justify
  it for $\ell_i$: either $\ell_i=0$, and we are done (because
  $\gamma(0)=\y$ is in $\mathscr{R}$), or $\ell_i>0$, so that $\ell_i$
  is in some interval $g_\ell$ (since then it does not belong to
  $h_i$), and thus $\gamma(\ell_i)$ is in $W_1 \subset \mathscr{R}$.
  Because $\overline{B_j}\cap\mathscr{R}$ is connected,
  $\gamma(\ell_i)$ and $\gamma(r_i)$ can be connected by a path in
  $\overline{B_j} \cap \mathscr{R}$, which is contained in $C \cap
  \mathscr{R}$.
\end{proof}


\section{Proof of the genericity properties}\label{sec:gen}

The algorithms of Subsections~\ref{ssec:canny} and~\ref{ssec:us} rely
on the fact that assumption $\BS$ holds in generic coordinates. We
discuss here the two cases we need (Lemmas~\ref{prop:3}
and~\ref{prop:4}) in the simplified case where $Q$ is empty (the
arguments carry over to the general cases). Thus, we let
$\f=f_1,\dots,f_p$ be a system that satisfies $\AS$; recall that
Lemma~\ref{prop:3} discusses $p$ arbitrary, and Lemma~\ref{prop:4} has
$p=1$.

In both cases, in generic coordinates, ${W_i}$ has dimension $i-1$ for
all $i=1,\dots,n-p$~\cite{BaGiHeMb97,BaGiHeMb01}. Then, points $(a)$
and $(b)$ of assumption $\BS$ are established in~\cite{SaSc03} when
$\sing(V)=\emptyset$. Since the assumption $\sing(V)=\emptyset$ was
only used to ensure that ${W_i}$ had dimension $i-1$, we obtain $(a)$
and $(b)$ in our case as well. Point~$(c)$ says that ${W_1}$ is
finite; this follows from the previous claim with $i=1$.  Point $(d)$,
which says that in generic coordinates $\crit(\Pi_1,{W_i})$ is finite,
is the most delicate of these properties; in the general case where
$p$ and $i$ are arbitrary, we do not know whether it always holds.

In Subsection~\ref{ssec:canny}, we have $p$ arbitrary and $i=2$: in
this case, ${W_2}$ is generically a curve in Noether position for
$\Pi_1$; this easily implies point~$(d)$, and thus finishes the proof
of Lemma~\ref{prop:3}. In Subsection~\ref{ssec:us}, for
Lemma~\ref{prop:4}, we need the case where $p=1$ and $i$ is arbitrary.
This turns out to be substantially harder; we sketch the proof in what
follows.

We work with the parameter space $\C^{i}\times \C^{ni}$; to an element
$(\g,\e)$ of $\C^{i}\times \C^{ni}$, with $\e=(\e_1,\dots,\e_i)$ and
all $\e_k$ in $\C^n$, we associate the linear maps
$$\begin{array}{rccc}
\Pi_\e:&\C^n& \to& \C^i \\
& \x=(x_1,\dots,x_n) & \mapsto& (\e_1\cdot \x,\dots,\e_i\cdot \x)
\end{array}
\qquad\text{and}\qquad
\begin{array}{rccc}
\rho_\g:&\C^i& \to& \C \\
& \y=(y_1,\dots,y_i) & \mapsto& \g \cdot \y.
\end{array}$$
We also define 
$W_\e=\crit(\Pi_\e,V)$. We will prove 
that for a generic $\e$, $\sing(W_\e)$  and 
$\crit(\rho_{\g_0} \circ \Pi_\e,W_\e)$ are
finite, with
$\g_0=(1,0,\dots,0)$. Changing the coordinates to bring $\e$ to the
first $i$ unit vectors gives  point $(d)$ of assumption $\BS$ for Lemma~\ref{prop:4}
(the last statement of this lemma is discussed hereafter).

\noindent For $\e \in \C^{ni}$ and $i+1 \le \ell \le n$, let $M_\ell$
be the $(i+1)$-minor built on columns $1,\dots,i,\ell$ of the matrix
$$\mM_\e=\left [~\begin{smallmatrix}
    \e^t_1 \\ \vdots \\ \e^t_i \\ \grad(f)
  \end{smallmatrix}~\right ].$$ We say that property ${\bf
  a}_1(\e)$ is satisfied if the following holds: $W_\e$ is the
zero-set of $(f,M_{i+1},\dots,M_n)$, the Jacobian matrix of
$(f,M_{i+1},\dots,M_n)$ has rank $n-i+1$ at all points of
$W_\e-\sing(V)$, $W_\e$ is $(i-1)$-equidimensional and $\sing(W_\e)$
is finite. Note that after changing coordinates to bring $\e$ to
the first $i$ unit vectors, this property implies the last claim of
Lemma~\ref{prop:4}.

For $0 \le j\le i$, define next $S_j = \{ \x \in\reg(V) \ | \
\dim(\Pi_\e(T_\x V)) = j\}.$ The sets $S_j$ form a partition of
$\reg(V)$; we say that property ${\bf a}_2(\e)$ is satisfied if for
$j=0,\dots,i$, $S_j$ is either empty or a non-singular constructible
subset of $\reg(V)$. If ${\bf a}_2(\e)$ holds, let
$m(n,i,j)=\max(0,\dim(S_j)-n+1+j)$ and $M(n,i,j)=\dim(S_j)$.  Then for
$m(n,i,j) \le \ell \le M(n,i,j)$, define finally
$$S_{j,\ell} = \{ \x \in S_j \ | \ \dim(\Pi_\e(T_\x S_j)) = \ell\}.$$
Under ${\bf a}_2(\e)$, the sets $S_{j,\ell}$ form a partition of
$S_j$. Then, property ${\bf a}_3(\e)$ holds if for $j=0,\dots,i$ and
$\ell=m(n,i,j),\dots,M(n,i,j)$, $S_{j,\ell}$ is either empty or a
non-singular constructible subset of $S_j$. The sets $S_j$ and
$S_{j,\ell}$ can be rewritten in terms of the standard notation of
Thom-Boardman strata~\cite{Thom55,Boardman67}. Hence, Mather's
transversality result for projections~\cite{Mather73,AlOt92,AlBaOt01}
implies the following lemma.
\begin{lemma}\label{lemma:a123}
  For a generic $\e$ in $\C^{ni}$, properties $\a_1(\e)$, $\a_2(\e)$
  and $\a_3(\e)$ are satisfied, and the inequality $\dim(S_{j,\ell})
  \le \ell$ holds for $\ell \le i-1$ and $m(n,i,j) \le \ell \le
  M(n,i,j)$.
\end{lemma}
Let now $\E=(\E_1,\dots,\E_i)$ be $ni$ indeterminates, that stand for
the vectors $\e=(\e_1,\dots,\e_i)$ and let ${\bf G}=(G_1,\dots,G_i)$
be indeterminates for $\g=(g_1,\dots,g_i)$. Let $J$ be the Jacobian
matrix of the polynomials $(f,M_{i+1},\dots,M_n)$, where we take
partial derivatives in the variables $\X$ only. Let further $\r$ be
the row vector of length $n$ given by
$$\r =\left [
  \begin{matrix}
G_1 & \cdots & G_i
  \end{matrix} \right ] \left [ \begin{matrix}
\E_1^t \\ \vdots \\ \E_i^t
\end{matrix} \right ],$$ and let finally $J'$ be the matrix obtained
by adjoining the row $\r$ to $J$. We define the algebraic set $X
\subset \C^i\times \C^{ni} \times \C^n$ as the set of all $(\g,\e,\x)
\in \C^i\times \C^{ni} \times \C^n$ such that
$f(\x)=M_{i+1}(\x,\e)=\cdots=M_n(\x,\e)=0$ and all $(n+2-i)$-minors of $J'(\g,\e,\x)$ vanish. 
Finally, we define the projections $\alpha: (\g,\e,\x) \mapsto
(\g,\e)$ and $\gamma: (\g,\e,\x) \mapsto \e$.

\begin{lemma}\label{prop:dimxe}
  If $\a_1(\e)$, $\a_2(\e)$ and $\a_3(\e)$ holds, then 
  $X \cap \gamma^{-1}(\e)$ has dimension at most $i$
\end{lemma}
\noindent Let finally $Y \subset \C^i\times\C^{ni}$ be the Zariski closure of
the set of all $(\g,\e) \in \C^i\times\C^{ni}$ such that the fiber $X
\cap \alpha^{-1}(\g,\e)$ is infinite. Lemma~\ref{prop:dimxe} is the
key to the following result.
\begin{lemma}\label{lemma:13}
  The set $Y$ is a strict algebraic subset of $\C^i\times\C^{ni}$ and for $(\g,\e)$ in $\C^i\times\C^{ni}-Y$, $\crit(\rho_\g \circ
  \Pi_\e, W_\e)$ is finite.
\end{lemma}

For any invertible $i\times i$ matrix $\mM$, the defining equations of
$X$ are multiplied by a non-zero constant through the change of
variables $({\bf G},\E,\X) \mapsto (\mM^{-1}{\bf G},\mM\E,\X)$, so $X$ is
stabilized by this action. Thus, a point $(\g,\e)$ in $\C^i\times
\C^{ni}$ belongs to $Y$ if and only if $(\mM^{-1}\g,\mM\e)$ does. One
deduces that all points of $Y$ locally look the same: since there
exists a point $(\g,\e)$ not in $Y$, and since $Y$ is closed, there
exists an open set $A \subset \{\g_0 \} \times \C^{ni}$ such that $A
\cap Y =\emptyset$, with $\g_0=(1,0,\dots,0)$; this is what we wanted.

\bibliographystyle{plain}
\bibliography{theo1}

\begin{thebibliography}{10}

\bibitem{AlBaOt01}
A.~Alzati, E.~Ballico, and G.~Ottaviani.
\newblock The theorem of {M}ather on generic projections for singular
  varieties.
\newblock {\em Geom. Dedicata}, 85(1-3):113--117, 2001.

\bibitem{AlOt92}
A.~Alzati and G.~Ottaviani.
\newblock The theorem of {M}ather on generic projections in the setting of
  algebraic geometry.
\newblock {\em Manuscripta Math.}, 74(4):391--412, 1992.

\bibitem{BaGiHeMb97}
B.~Bank, M.~Giusti, J.~Heintz, and G.-M. Mbakop.
\newblock Polar varieties and efficient real equation solving: the hypersurface
  case.
\newblock {\em Journal of Complexity}, 13(1):5--27, 1997.

\bibitem{BaGiHeMb01}
B.~Bank, M.~Giusti, J.~Heintz, and G.-M. Mbakop.
\newblock Polar varieties and efficient real elimination.
\newblock {\em Mathematische Zeitschrift}, 238(1):115--144, 2001.

\bibitem{BaPoRo96}
S.~Basu, R.~Pollack, and M.-F. Roy.
\newblock Computing roadmaps of semi-algebraic sets (extended abstract).
\newblock In {\em STOC}, pages 168--173. ACM, 1996.

\bibitem{BPRRoadmap}
S.~Basu, R.~Pollack, and M.-F. Roy.
\newblock Computing roadmaps of semi-algebraic sets on a variety.
\newblock {\em Journal of the AMS}, 3(1):55--82, 1999.

\bibitem{BaPoRo06}
S.~Basu, R.~Pollack, and M.-F. Roy.
\newblock {\em Algorithms in real algebraic geometry}, volume~10 of {\em
  Algorithms and Computation in Mathematics}.
\newblock Springer-Verlag, second edition, 2006.

\bibitem{Boardman67}
J.~M. Boardman.
\newblock Singularities of differentiable maps.
\newblock {\em Publ. Math. Inst. Hautes {\'E}tudes Sci.}, 33:21--57, 1967.

\bibitem{BrDaTrOk07}
J.-P. Brasselet, J.~Damon, {L.~D}. Trang, and M.~Oka, editors.
\newblock {\em Singularities in geometry and topology}.
\newblock World Scientific, 2007.

\bibitem{CannyThese}
J.~Canny.
\newblock {\em The complexity of robot motion planning}.
\newblock PhD thesis, MIT, 1987.

\bibitem{Canny}
J.~Canny.
\newblock Computing roadmaps in general semi-algebraic sets.
\newblock {\em The Computer Journal}, 36(5):504--514, 1993.

\bibitem{D5}
J.~{Della Dora}, C.~Discrescenzo, and D.~Duval.
\newblock About a new method method for computing in algebraic number fields.
\newblock In {\em EUROCAL 85 Vol. 2}, volume 204 of {\em LNCS}. Springer, 1985.

\bibitem{Eisenbud95}
D.~Eisenbud.
\newblock {\em Commutative algebra with a view toward algebraic geometry},
  volume 150 of {\em Graduate Texts in Mathematics}.
\newblock Springer-Verlag, 1995.

\bibitem{FiGiSm95}
N.~Fitchas, M.~Giusti, and F.~Smietanski.
\newblock Sur la complexit{\'e} du th{\'e}or{\`e}me des z{\'e}ros.
\newblock In {\em Approximation and Optimization in the Caribbean II}, volume~8
  of {\em Approximation and Optimization}, pages 247--329. Verlag Peter Lang,
  1995.

\bibitem{GiHe93}
M.~Giusti and J.~Heintz.
\newblock La d{\'e}termination des points isol{\'e}s et de la dimension d'une
  vari{\'e}t{\'e} alg{\'e}brique peut se faire en temps polynomial.
\newblock In {\em Computational Algebraic Geometry and Commutative Algebra},
  volume XXXIV of {\em Symposia Matematica}, pages 216--256. Cambridge
  University Press, 1993.

\bibitem{GRRoadmap}
L.~Gournay and J.-J. Risler.
\newblock Construction of roadmaps in semi-algebraic sets.
\newblock {\em Appl. Alg. Eng. Comm. Comp.}, 4(4):239--252, 1993.

\bibitem{Heintz83}
J.~Heintz.
\newblock Definability and fast quantifier elimination in algebraically closed
  fields.
\newblock {\em Theoret. Comput. Sci.}, 24(3):239--277, 1983.

\bibitem{HRSRoadmap}
J.~Heintz, M.-F. Roy, and P.~Solerno.
\newblock Single exponential path finding in semi-algebraic sets {II}: The
  general case.
\newblock In {\em Algebraic geometry and its applications, collections of
  papers from Abhyankar's 60-th birthday conference}. Purdue University,
  West-Lafayette, 1994.

\bibitem{HeSc80}
J.~Heintz and C.~P. Schnorr.
\newblock Testing polynomials which are easy to compute (extended abstract).
\newblock In {\em STOC}, pages 262--272. ACM, 1980.

\bibitem{Lecerf00}
G.~Lecerf.
\newblock Computing an equidimensional decomposition of an algebraic variety by
  means of geometric resolutions.
\newblock In {\em ISSAC'00}, pages 209--216. ACM, 2000.

\bibitem{Mather73}
J.~N. Mather.
\newblock Generic projections.
\newblock {\em Ann. of Math.}, 98:226--245, 1973.

\bibitem{Mumford76}
D.~Mumford.
\newblock {\em Algebraic Geometry I, Complex projective varieties}.
\newblock Classics in Mathematics. Springer Verlag, 1976.

\bibitem{Rouillier99}
F.~Rouillier.
\newblock Solving zero-dimensional systems through the {R}ational {U}nivariate
  {R}epresentation.
\newblock {\em Applicable Algebra in Engineering, Communication and Computing},
  9(5):433--461, 1999.

\bibitem{SaSc03}
M.~{Safey el Din} and \'E. Schost.
\newblock Polar varieties and computation of one point in each connected
  component of a smooth real algebraic set.
\newblock In {\em ISSAC'03}, pages 224--231. ACM, 2003.

\bibitem{Schost03}
{\'E}.~Schost.
\newblock Computing parametric geometric resolutions.
\newblock {\em Appl. Algebra Engrg. Comm. Comput.}, 13(5):349--393, 2003.

\bibitem{Sharir}
J.~Schwarz and M.~Sharir.
\newblock On the piano mover's problem {II}: General techniques for computing
  topological properties of real algebraic manifolds.
\newblock {\em Adv. Appl. Math.}, 4:298--351, 1983.

\bibitem{Shafarevich77}
I.~Shafarevich.
\newblock {\em Basic Algebraic Geometry 1}.
\newblock Springer Verlag, 1977.

\bibitem{Thom55}
R.~Thom.
\newblock Les singularit{\'e}s des applications diff{\'e}rentiables.
\newblock {\em Ann. Inst. Fourier}, 6:43--87, 1955-56.

\bibitem{Todd37}
{J. A.} Todd.
\newblock The arithmetical invariants of algebraic loci.
\newblock {\em Proc. Lond. Math. Soc.}, 43:190--225, 1937.

\bibitem{ZaSa58}
O.~Zariski and P.~Samuel.
\newblock {\em Commutative algebra}.
\newblock Van Nostrand, 1958.

\end{thebibliography}

\newpage


\section*{Appendix}

We give the proofs of several of the results announced before; we
usually do not repeat the necessary definitions, so we indicate to
which page the reader should refer. We mostly follow the order in
which the statements are made in the text; in a few cases, we modify
the order to avoid excessive cross-referencing.


\subsection*{Completion of the proof of Theorem~\ref{theo:big} on page~\pageref{theo:big}}

\begin{lemma}
  Let $\F=(f_1,\dots,f_p)$ be a system that satisfies assumption
  $\AS$, and let $i \le n-p$. For all $\x=(x_1,\dots,x_{i-1})$ in
  $\R^{i-1}$, the system
  $\f_x=(f_1,\dots,f_p,X_1-x_1,\dots,X_{i-1}-x_{i-1})$ satisfies the
  following properties:
  \begin{itemize}
  \item the ideal $I_\x=\langle \f_\x \rangle$ is radical;
\item the variety $V_\x$ it defines is equidimensional of dimension
  $n-p-(i-1)$;
\item $\sing(V_\x)$ is finite;
\item $V_\x \cap \R^n$ is bounded.
\end{itemize}
Besides, $V_\x$ intersects $W_i=\crit(\Pi_i,V(\F))$ in finitely many points.
\end{lemma}
\begin{proof}
  Remark that $V_\x$ is either empty or of dimension at least
  $n-p-(i-1)$, by Krull's theorem. Let us show that it is not empty:
  since $V=V(\F)$ is in Noether position for $\Pi_d$, for any
  $\x'=(x_1,\dots,x_d)$, $\Pi_d^{-1}(\x')\cap V$ is not empty. A
  fortiori, $\Pi_{i-1}^{-1}(\x)\cap V$ is not empty, and thus all
  irreducible components of $V_\x$ have dimension at least
  $n-p-(i-1)$.

  Let $\y$ be in $V_\x$. By construction, the Jacobian of
  $(f_1,\dots,f_p,X_1-x_1,\dots,X_{i-1}-x_{i-1})$ has full rank if and
  only if $\y$ is in $W_i=W_i \cup \sing(V)$. However, since $W_i$
  is in Noether position for $\Pi_{i-1}$, $W_i \cap
  \Pi_{i-1}^{-1}(\x)$ is finite (which gives the last
  assertion). Since $\sing(V)$ is finite as well, and since $n-p-(i-1)
  \ge 1$, each irreducible component of $V_\x$ contains a point $\y$
  where the former Jacobian matrix has full rank. Consequently, we
  deduce that each irreducible component of $I_\x$ has dimension
  $n-p-(i-1)$ (by the Jacobian criterion) and that $I_\x$ is radical
  (by Macaulay's unmixedness theorem). We have thus established the
  first two points.

  As a consequence, the singular points of $V_\x$ are the points where
  the rank of the former Jacobian drops; as we have seen, they are in
  $W_i \cap \Pi_{i-1}^{-1}(\x)$, and thus in finite number. This
  gives the third point. The next point is obvious, since $V_\x \cap
  \R^n \subset V \cap \R^n$, and the latter is bounded.
\end{proof}

We can now complete the proof of Theorem~\ref{theo:big}. We start by
proving that $\mathscr{C}'\cup W_i$ is a $d'$-roadmap of
$(V,\mathscr{P})$. The connectivity property ${\rm RM}'_1$ is
established in Section~\ref{sec:proofbig}. Property ${\rm RM}'_2$ is
clear from the construction. Next, the dimension of $W_i$ is at most
$i-1$ by point $(b)$ of $\BS$. We have seen in the previous lemma that
all fibers $\Pi_{i-1}^{-1}\cap V$ have dimension $n-p-(i-1)$; because
$\mathscr{C}$ is finite by assumption $\BS$, this implies that
$\dim(\mathscr{C}') = n-p-(i-1)$, and thus that $\dim(\mathscr{R}) =
d'$. Thus, we have ${\rm RM}'_3$.  Finally, by construction,
$\mathscr{P}$ is contained in $\mathscr{C}$, so obtain ${\rm RM}'_4$.

The last propriety we need is that $\mathscr{C}'\cap W_i$ has
dimension at most zero: this is the last assertion of the previous
lemma.


\subsection*{Proof of Lemma~\ref{lemma:glue} on page~\pageref{lemma:glue}}

We prove the following claim: {\em Suppose that $\mathscr{R}_1 \cup
  \mathscr{R}_2$ is a $j$-roadmap of $(V,\mathscr{P})$, with
  $\mathscr{R}_1 \cap \mathscr{R}_2$ finite. Let $\mathscr{R}'_1$ and
  $\mathscr{R}'_2$ be roadmaps of respectively $(\mathscr{R}_1,
  (\mathscr{R}_1 \cap \mathscr{R}_2) \cup \mathscr{P})$ and
  $(\mathscr{R}_2, (\mathscr{R}_1 \cap \mathscr{R}_2) \cup
  \mathscr{P})$. Then $\mathscr{R}'_1 \cup \mathscr{R}'_2$ is a
  roadmap of $(V,\mathscr{P})$.}

\begin{lemma}\label{lem:roadCC}
  If $\mathscr{R}$ is an $i$-roadmap of $V$, then for each connected
  component $C$ of $V \cap \R^n$, $C \cap \mathscr{R}$ is a connected
  component of $\mathscr{R}\cap\R^n$.
\end{lemma}
\begin{proof}
  We know that $C \cap \mathscr{R}$ is connected. Besides, $C$ is both
  open and closed in $V \cap \R^n$, so that $C\cap \mathscr{R}$ is
  open and closed in $\mathscr{R}\cap \R^n$.
\end{proof}

\begin{lemma}\label{lem:prelim}
  If $\mathscr{R}$ is an $i$-roadmap of $V$ and if $\mathscr{R}'$ is a
  $j$-roadmap of $\mathscr{R}$ then $\mathscr{R}'$ is a $j$-roadmap of~$V$.
\end{lemma}
\begin{proof}
  Since the dimension of $\mathscr{R}'$ is $j$, it is sufficient to
  prove that for each connected component $C$ of $V \cap \R^n$, $C
  \cap \mathscr{R}'$ is non empty and connected.  Since $\mathscr{R}$
  is a roadmap of $V$, $C\cap \mathscr{R}$ is a connected component of
  $\mathscr{R}\cap\R^n$ (Lemma~\ref{lem:roadCC}). Since $\mathscr{R}'$
  is a roadmap of $\mathscr{R}$, $C\cap
  \mathscr{R}\cap\mathscr{R}'=C\cap \mathscr{R}'$ is a connected
  component of $\mathscr{R}'\cap\R^n$.
\end{proof}

\noindent We can now prove our claim. By Lemma \ref{lem:prelim}, it is
sufficient to prove that $\mathscr{R}'_1\cup\mathscr{R}'_2$ is a
roadmap of $\mathscr{R}_1\cup\mathscr{R}_2$. Let $C$ be a connected
component of $\mathscr{R}_1 \cup \mathscr{R}_2$. First, we prove that
$C\cap (\mathscr{R}'_1\cup\mathscr{R}'_2)$ is not empty. Indeed, $C$
contains a connected component of either $\mathscr{R}_1$ or
$\mathscr{R}_2$ (since it contains a point of say $\mathscr{R}_1$, it
contains its connected component); and as such, $C$ intersects either
$\mathscr{R}'_1$ or $\mathscr{R}'_2$.

We prove now that $C\cap (\mathscr{R}'_1\cup\mathscr{R}'_2)$ is
connected. Consider a couple of points $\x, \y$ in $C\cap
(\mathscr{R}'_1\cup\mathscr{R}'_2 )$. Since $C$ is connected, there
exists a continuous path $\gamma:[0, 1]\rightarrow C$ such that
$\gamma(0)=\x$ and $\gamma(1)=\y$.  Since
$\mathscr{R}_1\cap\mathscr{R}_2$ is finite, we can reparametrize
$\gamma$, to ensure that $\gamma^{-1}(\mathscr{R}_1\cap\mathscr{R}_2)$
is finite. Denote by $t_1<\cdots< t_r$ the set
$\gamma^{-1}(\mathscr{R}_1\cap\mathscr{R}_2)$ and let $t_0=0$ and
$t_{r+1}=1$.  Then, we replace $\gamma$ by a continuous path $\gamma'$
defined on the segments $[t_i,t_{i+1}]$ as follows:
  \begin{itemize}
  \item For $1 \le i < r$, $\gamma((t_i, t_{i+1}))$ is connected and
    contained in $\mathscr{R}_1\cup\mathscr{R}_2 -
    \mathscr{R}_1\cap\mathscr{R}_2$, so it is contained in (say)
    $\mathscr{R}_1$. By continuity, $\gamma([t_i, t_{i+1}])$ is
    contained in $\mathscr{R}_1$, and thus actually in a connected
    component $C_i$ of $\mathscr{R}_1$, with $C_i \subset C$.  Both
    $\gamma(t_i)$ and $\gamma(t_{i+1})$ are in
    $\mathscr{R}_1\cap\mathscr{R}_2$, and thus in
    $\mathscr{R}'_1\cap\mathscr{R}'_2$, and in particular in
    $\mathscr{R}'_1$. Since by definition $C_i \cap \mathscr{R}'_1$ is
    connected, there exists a continuous semi-algebraic path
    $\gamma':[t_i,t_{i+1}] \to C_i \cap \mathscr{R}'_1$ with
    $\gamma'(t_i)=\gamma(t_i)$ and $\gamma'(t_{i+1})=\gamma(t_{i+1})$.
 
  \item The case $i=0$ needs to be taken care of only if $t_0 < t_1$,
    so that $\x=\gamma(t_0)$ is either in $\mathscr{R}_1$ or in
    $\mathscr{R}_2$, but not in both.  As before, we start by
    remarking that $\gamma([t_0,t_1])$ is contained in a connected
    component $C_0$ of say $\mathscr{R}_1$, with $C_0 \subset C$. This
    implies that $\x=\gamma(t_0)$ is in $\mathscr{R}_1$; since $\x$ is
    in $\mathscr{R}'_1\cup\mathscr{R}'_2$, it is actually in
    $\mathscr{R}'_1$. As before, $\gamma(t_1)$ is in $\mathscr{R}'_1$,
    and the conclusion follows as in the previous case. The case $i=r$
    is dealt with similarly.
  \end{itemize}


\subsection*{Proof of Lemma~\ref{prop:wibar} on page~\pageref{prop:wibar}} 

The following lemma is similar to Proposition~7.3 in~\cite{BaPoRo06};
the proof is a consequence of the semi-algebraic implicit function
theorem. Hereafter, the closure notation $\overline B$ refers to the
closure for the Euclidean topology.
\begin{lemma}\label{7.3}
  Let $\x$ be in $Z\cap \R^n-W_1$ and let $x_1=\Pi_1(\x)$. There
  exists an open, semi-algebraic, connected neighborhood $X(\x)$ of
  $\x$ such that $X(\x) \cap Z_{< x_1}$ is non-empty and connected,
  and $X(\x) \cap Z_{x_1}$ is contained in $\overline {X(\x) \cap Z_{<
      x_1}}$.
\end{lemma}

\begin{lemma}\label{lemma:10}
  Let $\y$ be in $Z \cap \R^n-W_1$ and let $y_1=\Pi_1(\y)$. There
  exists a unique connected component $B(\y)$ of $Z_{<y_1}$ such that
  $X(\y) \cap Z_{<y_1} \subset B(\y)$. Besides, $B(\y)$ is the unique
  connected component of $Z_{<y_1}$ such that $\y$ is in
  $\overline{B(\y)}$.
\end{lemma}
\begin{proof}
  Because $X(\y) \cap Z_{<y_1}$ is non-empty and connected
  (Lemma~\ref{7.3}), it is contained in a connected component $B(\y)$
  of $Z_{<y_1}$. The connected components of $Z_{<y_1}$ are pairwise
  disjoint, so $B(\y)$ is well-defined. By Lemma~\ref{7.3} again, $\y$
  is in $\overline{X(\y) \cap Z_{<y_1}}$, and thus in
  $\overline{B(\y)}$. Suppose finally that $\y$ is in $\overline{B'}$,
  for another connected component $B'$ of $Z_{<y_1}$. Then, there
  exists a point of $B'$ in $X(\y)$, because $X(\y)$ is open. This
  point is in $X(\y) \cap Z_{<y_1}$, and thus in $B(\y)$ as well, a
  contradiction.
\end{proof}

\begin{lemma}\label{lemma:cont}
  Let $\y$ be in $Z \cap \R^n-W_1$ and let $y_1=\Pi_1(\y)$. For $\y'$
  in $X(\y) \cap (Z_{y_1}\cap \R^n-W_1)$, we have $B(\y')=B(\y)$.
\end{lemma}
\begin{proof} The reasoning is the same as in the previous lemma. We
  know that $\y'$ is in $\overline{B(\y')}$. Since $\y'$ is in $X(\y)$
  and $X(\y)$ is open, there exists a point of $B(\y')$ in $X(\y) \cap
  Z_{<y_1}$. This point is in $B(\y)$ as well, so $B(\y')=B(\y)$.
\end{proof}

\begin{lemma}\label{lemma:12}
  Let $x$ be in $\R$ and let $\gamma$ be a continuous semi-algebraic
  map $A \to Z_x-W_1$, where $A \subset\R^k$ is a connected
  set. Then, there exists a unique connected component $B$ of $Z_{<x}$
  such that for all $\a\in A$, $\gamma(\a) \in \overline{B}$.
\end{lemma}
\begin{proof}
  By Lemma~\ref{lemma:cont}, the map $\a \mapsto B(\gamma(\a))$ is
  locally constant, so it is constant. So, with $B=B({\gamma(\a_0)})$,
  for some $\a_0$ in $A$, we have $B({\gamma(\a)})=B$ for all $\a$ in
  $A$, and thus $\gamma(\a) \in \overline{B}$ for all $\a$ by
  Lemma~\ref{lemma:10}. Uniqueness is a consequence of the second part
  of Lemma~\ref{lemma:10}.
\end{proof}

\noindent We can now prove Lemma~\ref{prop:wibar}.  Let thus $\gamma$
be a continuous semi-algebraic map $A\to Z_{\le x}-Z_x\cap W_1$, where
$A \subset \R^k$ is a connected semi-algebraic set; we prove that
$\gamma(A)$ is contained in the closure $\overline B$ of a connected
component $B$ of $Z_{<x}$. If $\gamma(A)$ is contained in $Z_{< x}$,
then, since it is connected, it is contained in a uniquely defined
connected component $B$ of $Z_{<x}$, and we are done.

Else, let $G=\gamma^{-1}(Z_x)$, which is closed in $A$. We decompose
it into its connected components $G_1,\dots,G_N$. Because all $G_i$
are closed in $G$, they are closed in $A$. Let also $H_1,\dots,H_M$ be
the connected components of $A-G$; hence, the $H_j$ are open in $A$
(because they are open in $A-G$, which is open in $A$).  The sets
$G_i$ and $H_j$ form a partition of $A$; we assign them some connected
components of $Z_{<x}$.
\begin{itemize}
\item Since $G_i$ is connected and $\gamma(G_i)$ is contained in
  $Z_x-W_1$, Lemma~\ref{lemma:12} shows that there exists a unique
  connected component $B(G_i)$ of $Z_{<x}$ such that for all $\g$ in
  $G_i$, $\gamma(\g) \in \overline{B(G_i)}$.
\item Since $H_j$ is connected and $\gamma(H_j)$ is contained in
  $Z_{<x}$, there exists a unique connected component $B(H_j)$ of
  $Z_{<x}$ that contains $\gamma(H_j)$. Since $\gamma$ is continuous,
  for all $\h$ in the closure $\overline{H_j}$ of $H_j$ in $A$, we
  still have $\gamma(\h) \in \overline{B(H_j)}$.
\end{itemize}
Since the sets $G_i$ and $H_j$ form a partition of~$A$, we deduce from
the previous construction a function $\a \mapsto B(\a)$ in the obvious
manner: if $\a$ is in $G_i$, we let $B(\a)=B(G_i)$; if $\a$ is in
$H_j$, we let $B(\a)=B(H_j)$. It remains to prove that this function
is constant on $G$; then, if we let $B$ be the common value $B(\a)$,
for all $\a$ in $G$, $\gamma(\a)$ is in $\overline{B}$ by construction
(uniqueness is clear).  To do so, it is sufficient to prove that for
any $\a$ in $A$, there exists a neighborhood $N(\a)$ of $\a$ such that
for all $\a'$ in $N(\a)$, $B(\a)=B(\a')$.
\begin{itemize}
\item If $\a$ is in some $H_j$, we are done, since $H_j$ is open, and
  $\a \mapsto B(\a)$ is constant on $H_j$.
\item Else, $\a$ is in some $G_i$. Remark that $\a$ is the closure of
  no other $G_{i'}$, since the $G_i$ are closed; however, $\a$ can
  belong to the closure of some $H_j$. For definiteness, let $J$ be
  the set of indices such that $\a$ is in $\overline{H_j}$ for $j$ in
  $J$, and let $e>0$ be such that the open ball $B(\a,e)$ intersects
  no $G_{i'}$, for $i'\ne i$, and no $\overline{H_j}$, for $j$ not in
  $J$. Since $\a$ is in $G_i$, we know that $\gamma(\a)$ is in
  $\overline{B(G_i)}$; for $j$ in $J$, since $\a$ is in
  $\overline{H_j}$, we also have that $\gamma(\a)$ is in
  $\overline{B(H_j)}$. However, since $\gamma(\a)$ is in $Z_x-W_1$,
  the second statement in Lemma~\ref{lemma:10} implies that
  $B(G_i)=B(H_j)$. Since every $\a'$ in $B(\a,e)$ is either in $G_i$ or
  in some $H_j$ with $j$ in $J$, we are done.
\end{itemize}

\noindent This concludes the proof of Lemma~\ref{prop:wibar}. The following
corollary will be used to prove Lemma~\ref{prop:ret}.
\begin{corollary}\label{coro:stronger}
  Let $x$ be in $\R$ such that $Z_x \cap W_1=\emptyset$ and let $C$
  be a connected component of $Z_{\le x}$. Then if $C_{<x}$ is
  non-empty, it is connected.
\end{corollary}
\begin{proof}
  Consider the inclusion map $C \to Z_{\le x}$. Since $Z_x \cap W_1$
  is empty, this map satisfies the assumptions of
  Proposition~\ref{prop:wibar}; this implies that there exists a
  unique connected component $B$ of $Z_{<x}$ such that $C \subset
  \overline{B}$.  This equality implies that $C_{<x}$ is contained in
  $\overline{B}_{<x}$; one easily checks that $B=\overline{B}_{<x}$,
  so that $C_{<x} \subset B$.  Now, let $B'$ be a connected component
  of $C_{<x}$, so that $B'$ is actually a connected component of
  $Z_{<x}$. The inclusion $B' \subset C_{<x}$ implies $B'\subset
  C_{<x} \subset B$ and thus $B'=C_{<x}=B$. Since $B$ is connected,
  $C_{<x}$ is, as claimed.
\end{proof}


\subsection*{Proof of Lemma~\ref{prop:ret} on page~\pageref{prop:ret}}

Lemma~\ref{prop:ret} is a by-product of the following result.
\begin{lemma}\label{lemma:ret}
  Let $v<w$ be in $\R$ and let $A \subset (-\infty,w)\times \R^{n-1}$
  be a connected, bounded semi-algebraic set such that $A_{(v,w)}$ is
  a non-empty, smooth manifold, closed in $(v,w) \times \R^{n-1}$ and
  such that $\Pi_1$ is a submersion on $A_{(v,w)}$. Then, for all $x$
  in $[v,w)$, $A_{\le x}$ is non-empty and connected.
\end{lemma}
First, we deduce Lemma~\ref{prop:ret} from Lemma~\ref{lemma:ret}. Let
$C$ be a connected component of $Z_{\le w}$ and recall that we want to
prove that for $x$ in $[v,w]$, $C_{\le x}$ is a connected component of
$Z_{\le x}$; of course, we can assume that $x<w$. Then, it suffices to
prove that $C_{\le x}$ is non-empty and connected; then it is a easily
seen to be a connected component of $Z_{\le x}$. If $C_{(v,w]}$ is
empty, then for $x$ in $[v,w)$, $C_{\le x}=C$, so we are done. Hence,
we assume that $C_{(v,w]}$ is non empty.

We verify here that all assumptions of Lemma~\ref{lemma:ret} are
satisfied, with $A=C_{< w}$. Since $C_{(v,w]}$ is non empty and $Z_w
\cap W_1$ is empty, $C_{(v,w)}$ is non-empty: either there is a point
in $C_{(v,w)}$, or there is a point in $C_w$; this point is not in
$W_1$, so the implicit function theorem shows that $C_{(v,w)}$ is not
empty in this case as well.  Besides, since $Z_w \cap W_1$ is empty,
by Corollary~\ref{coro:stronger}, $C_{<w}$ is connected.

To summarize, $C_{<w}$ is a connected and bounded semi-algebraic set;
$C_{(v,w)}$ is smooth and of positive dimension (because there is no
point in $W_1$ in $C_{(v,w)}$), closed in $(v,w)\times \R^{n-1}$
(because $C_{(v,w)} = C \cap ((v,w)\times \R^{n-1})$ and $C$ is
closed). Besides, we claim that $\Pi_1$ is a submersion on
$C_{(v,w)}$. First, remark that any point $\x$ of $C_{(v,w)}$, $T_\x
C_{(v,w)} = T_\x Z\cap \R^n$.  Since $\dim(Z) > 0$, and since there is
no point of $W_1$ on $Z_{(v,w)}$, we know that $\Pi_1(T_\x Z)=\C$,
which implies that $\Pi_1(T_\x Z\cap \R^n) =\R$. This establishes that
$\Pi_1$ is a submersion on $C_{(v,w)}$. We can thus apply
Lemma~\ref{lemma:ret}, which implies that $C_{\le x}$ is non-empty and
connected, as requested.

\medskip\noindent
Hence, we are left to prove Lemma~\ref{lemma:ret}. Let us first check
that $\Pi_1: A_{(v,w)} \to (v,w)$ is a proper mapping for the topology
induced by the Euclidean topology. By assumption, there exists a
closed set $X\subset \R^n$ such that $A_{(v,w)}=X \cap ((v,w)\times
\R^{n-1})$; since $A$ is bounded, we can take $X$ bounded as well. Let
$K$ be a compact set in $(v,w)$, so that $K$ is compact in $\R$
too. Then, $\Pi_1^{-1}(K)\cap A_{(v,w)}= X \cap (K\times \R^{n-1})$,
which is compact in $\R^n$, and thus in $A_{(v,w)}$. So $\Pi_1:
A_{(v,w)} \to (v,w)$ is proper.

Let $\zeta \in (v,w)$ be such that $A_{\zeta}$ is not empty (such a
$\zeta$ exists by assumption).  We apply Ehresmann's fibration
theorem~\cite[Th.~3.4]{BrDaTrOk07} to the projection $\Pi_1$ (which is
a proper submersion on $A_{(v,w)}$); this gives us a smooth
diffeomorphism of the form
$$\begin{array}{rrcl}\Psi: & A_{(v,w)} & \to & (v,w) \times A'_{\zeta} \\
  & (\alpha,\a) & \mapsto & (\alpha,\psi(\alpha,\a)), \end{array}$$               
where $A'_\zeta \subset \R^{n-1}$ is the set $\{(x_2,\dots,x_n) \ | \ (\zeta,x_2,\dots,x_n)\in A_\zeta\}$
(recall that $A_\zeta$ lies in $\R^n$). For the whole length of this proof,
vectors of the form $(\alpha,\a)$ have $\alpha$ in $\R$ and $\a$ in $\R^{n-1}$.

We use $\Psi$ to show that for $v< x < w$, $A_{\le x}$ is non-empty
and connected.  Let thus $x$ be fixed in $(v,w)$, and let $(\zeta,\z)$
be in $A_\zeta$.  Remark that $\Psi^{-1}(x, \z)$ is in $A_x$, proving
that $A_{\le x}$ is non-empty. To prove connectedness, we use a
similar process. Let $\y_0$ and $\y_1$ be in $A_{\le x}$. Since $A$ is
connected, there exists a continuous path $\gamma: [0,1] \to A$, with
$\gamma(t)=(\alpha(t),\a(t))$, that connects them. Let us replace
$\gamma$ by the path $g$ defined as follows:
  \begin{itemize}
  \item $g(t)=\gamma(t)$ if $\alpha(t) \le x$;
  \item $g(t)=\Psi^{-1}(x, \psi(\alpha(t),\a(t)))$ if $a(t) \ge x$.
  \end{itemize}
  The path $g(t)$ is well-defined, lies in $A_{\le x}$ by
  construction, and connect $\y_0$ to $\y_1$. This establishes our
  connectivity claim.

  Now, we can deal with the situation above $v$. We cannot directly
  use the fibration above, since it is not defined above $v$; instead,
  we will use a limiting process, that will rely on
  semi-algebraicity. To do so, we use a semi-algebraic fibration.
  Applying Hardt's semi-algebraic triviality theorem to the projection
  $\Pi_1$ on the semi-algebraic set $A_{<w}$ proves that there exist
  $z_0=-\infty < z_1 <\dots<z_m=w$ in $\R\cup \{-\infty\}$ such that
  above each interval $]z_i,z_{i+1}[$, there exists a semi-algebraic
  homeomorphism of the form
  $$\begin{array}{rrcl}\Phi_i: & A_{(z_i,z_{i+1})} & \to & (z_i,z_{i+1}) \times A'_{\rho_i} \\ 
    & (\alpha,\a) & \mapsto & (\alpha,\phi_i(\alpha,\a)). \end{array},$$               
  where $\rho_i$ is (for instance) $(z_i+z_{i+1})/2$
  and $A'_{\rho_i} \subset \R^{n-1}$ is 
  $\{(x_2,\dots,x_n) \ | \ (\rho_i,x_2,\dots,x_n)\in A_{\rho_i}\}$.

  Let $i_0$ be such that $v$ is in $[z_{i_0},z_{i_0+1})$ (so $v$ can
  be an interior point, or coincide with $z_{i_0}$). To prove that
  $A_{\le v}$ is non-empty, we actually prove that $A_v$ is. Let
  $\r_{i_0}$ be such that $(\rho_{i_0},\r_{i_0})$ is in
  $A_{\rho_{i_0}}$ (such a point exists, because $A_{\rho_{i_0}}$ is
  not empty, by the previous paragraphs). We define the function
  $\gamma: [0,1)\to A_{\le x}$ by $\gamma(t) =
  \Phi_{i_0}^{-1}(tv+(1-t)\rho_i, (\rho_i,\r_i))$. This is a
  semi-algebraic, continuous, bounded function, so it can be extended
  by continuity at $t=1$~\cite[Proposition~3.18]{BaPoRo06}. Since
  $\gamma(t)$ is in $A_{[v,\rho_i]}$ for $t<1$, $\gamma(1)$ is in
  $A_{[v,\rho_i]}$ too; besides, $\Pi_1(\gamma(t))= tv+(1-t)y_i$ for
  $t<1$, so $\Pi_1(\gamma(1))=v$. Hence, $\gamma(1)$ is in $A_v$, as
  requested.

  It remains to prove that $A_{\le v}$ is connected. Let thus $\y_0$
  and $\y_1$ be two points in $A_{\le v}$. Since $A_{\le \rho_i}$ is
  connected (first part of the proof) and semi-algebraic, $\y_0$ and
  $\y_1$ can be connected by a semi-algebraic path $\gamma$ in $A_{\le
    \rho_i}$. As we did previously, we replace $\gamma$ by a better
  path~$g$. Let $\varepsilon$ be an infinitesimal, let $A'$ be the
  extension of $A$ over $\R\langle \varepsilon\rangle$ and let $g$ be
  the path $[0,1]\subset \R\langle\varepsilon\rangle \to A'_{(v,w)}$
  be defined as follows (where as before
  $\gamma(t)=(\alpha(t),\a(t))$)
  \begin{itemize}
  \item $g(t)=\gamma(t)$ if $\alpha(t) \le v+\varepsilon$;
  \item $g(t)=\Phi_i^{-1}(v+\varepsilon, \phi_i(\alpha(t),\a(t)))$ if $\alpha(t) \ge v+\varepsilon$.
  \end{itemize}
  Obviously, $g$ is well-defined (since $\gamma$ has its image in
  $A_{\le \rho_i}$) and continuous, bounded over $\R$ and
  semi-algebraic. Its image $G$ is thus a connected semi-algebraic
  set, contained in $A'_{\le v+\varepsilon}$.  Let
  $G_0=\lim_\varepsilon G$. By construction, $\y_0$ and $\y_1$ are in
  $G_0$, $G_0$ is contained in $A_{\le v}$ and
  by~\cite[Proposition~12.43]{BaPoRo06}, $G_0$ is semi-algebraically
  connected. Our claim follows.


\subsection*{Statement and proof of Lemma~\ref{7} used on page~\pageref{Prop:Px2}}

\begin{lemma}\label{7}
  If ${\bf P}(x')$ holds for $x' < x$, then for $i \le r$,
  $\overline{B_i} \cap \mathscr{R}$ is non-empty and connected.
\end{lemma}

Let $B$ be one of the connected components $B_i$ of $C_{<x}$.  Since
$B$ is actually a connected component of $V_{<x}$ and $V\cap \R^n$ is
compact, $B$ contains a point of $W_1$ (the minimal point for
$\Pi_1$).  Hence, $B \cap \mathscr{R}$, and thus $\overline{B}\cap
\mathscr{R}$, are not empty. Next, we prove that any point $\y$ in
$\overline{B}\cap \mathscr{R}$ can be connected to a point $\z$ in $B
\cap \mathscr{R}$ by a path in $\overline{B}\cap \mathscr{R}$. Let us
first justify that this is sufficient to establish the lemma.

Consider two points $\y,\y'$ in $\overline{B}\cap \mathscr{R}$ and
suppose that they can be connected to some points $\z,\z'$ in $B \cap
\mathscr{R}$ by paths in $\overline{B}\cap \mathscr{R}$. Since $\z$
and $\z'$ are in $B$, they can be connected by a path $\gamma: [0,1]
\to B$. Let $x'=\max(\Pi_1(\gamma(t)))$, for $t$ in $[0,1]$; $x'$ is
well defined by the continuity of $\gamma$, and satisfies $x' <
x$. Then, both $\z$ and $\z'$ are in $B_{\le x'}$, and they can be
connected by a path in $B_{\le x'}$; hence, they are in the same
connected component $B'$ of $B_{\le x'}$. Now, $B'$ is a connected
component of $V_{\le x'}$, which implies by property ${\bf P}(x')$
that $B' \cap \mathscr{R}$ is connected. Hence, $\z$ and $\z'$, which
are in $B' \cap \mathscr{R}$, can be connected by a semi-algebraic
path in $B' \cap \mathscr{R}$, and thus within $B \cap
\mathscr{R}$. Summarizing, this proves that $\y$ and $\y'$ can be
connected by a path in $\overline{B}\cap \mathscr{R}$, as requested.

We are thus left to prove the claim made in the first
paragraph. Recall that $\mathscr{R}$ is the union of $W_i$ and of
$\mathscr{C}'=V \cap \Pi_{i-1}^{-1}(\Pi_{i-1}(\mathscr{C}))$, where
$\mathscr{C}=W_1 \cup \crit(\Pi_1,W_i) \cup \mathscr{P}$. We first
deal with points $\y$ in $\overline{B}\cap \mathscr{C}'$, and in a
second time with points $\y$ in $\overline{B} \cap (W_i -
\mathscr{C}')$.

\medskip\noindent{\bf Case 1.}  Let $\y$ be in $\overline{B}\cap
\mathscr{C}'$. We can assume that $\y$ is not in $B$, since for $\y$
in $B$ we can take $\z=\y$; since $\y$ is not in $B$, $\Pi_1(\y)=x$.

Since $B$ is semi-algebraic, by the curve selection lemma, there
exists a continuous semi-algebraic map $f:[0,1] \to \R^n$, with
$f(0)=\y$ and $f(t) \in B$ for $t$ in $(0,1]$.  Let $\varepsilon$ be a
new infinitesimal and let $\RR=\R\langle\varepsilon\rangle$; we let
$\varphi=(\varphi_1,\dots,\varphi_n) \in {\RR}^n$ be the
semi-algebraic germ of $f$ at $0$, so that $\lim_\varepsilon
\varphi=\y$. We consider the semi-algebraic set $\BB\subset {\RR}^n$
defined by
\begin{eqnarray*}\BB&=&\{ \x \in {\RR}^n \ |
  \ \x\in {\rm ext}(B,\RR) \text{~and~} (x_1,\dots,x_{i-1})=(\varphi_1,\dots,\varphi_{i-1})\},
\end{eqnarray*}
where ${\rm ext}$ denotes the extension to $\RR$. Since for all $t$ in
$(0,1]$, $f(t)$ is in $B$, $\varphi$ is in ${\rm ext}(B,\RR)$
by~\cite[Prop. 3.16]{BaPoRo06}, so that $\varphi$ is in~$\BB$;
in particular, this proves that $\y$ is in $\lim_\varepsilon
\BB$. Remark also that $\BB$ is bounded by an element of $\R$, and
that any point in $\lim_\varepsilon \BB$ is in $\overline{B} \cap
\Pi_{i-1}^{-1}(\Pi_{i-1}(\y))$.

Let $\BB_1,\dots,\BB_s \subset \RR^n$ be the semi-algebraically
connected components of $\BB$ (which are well-defined because $\BB$ is
not empty); hence, the $\BB_i$ are semi-algebraic sets. Because $\y$ is in
$\lim_\varepsilon (\BB)$, we can assume that it is in
$\lim_\varepsilon \BB_1$. Next, since $B$ is a semi-algebraically
connected component of $V_{< x}$, by~\cite[Prop.~5.24]{BaPoRo06},
$\BB_1$ is a semi-algebraically connected component of
${\rm ext}(V,\RR) \cap \Pi_{i-1}^{-1}(\varphi_1,\dots,\varphi_{i-1}).$
By the semi-algebraic implicit function theorem, this implies that
there exists a point $\psi$ in $\BB_1 \cap \crit(\Pi_,{\rm ext}
(V,\RR))$. Since polar varieties are defined by suitable Jacobian
minors, this means that $\psi$ is in $\BB_1 \cap {\rm ext}(W_i,\RR)$.
Because $\psi$ is in $\BB_1$, it is in ${\rm ext}(B,\RR)$, and
thus in ${\rm ext}(B \cap W_i, \RR)$. 

Let $\w=\lim_\varepsilon \psi$ and let $g$ be a representative of
$\psi$, so that $g(0)=\w$. By~\cite[Prop.~3.16]{BaPoRo06}, there
exists $t_0>0$ such that for all $t$ in $(0,t_0)$, $g(t)$ is in $B
\cap W_i$, which is contained in $\overline{B} \cap
\mathscr{R}$. Defining $\z=g(t_0/2)$, we see that $\z$ and $\w$ are
connected by a path in $\overline{B} \cap \mathscr{R}$.

Let $B_1=\lim_\varepsilon \BB_1$. Because $\BB_1$ is semi-algebraic,
bounded over $\R$ and semi-algebraically connected, $B_1$ is closed,
semi-algebraic and connected~\cite[Prop.~12.43]{BaPoRo06}. Besides, we
have seen above that it is contained in $\overline{B} \cap
\Pi_{i-1}^{-1}(\Pi_{i-1}(\y))$. Finally, it contains both $\y$ and
$\w$. Hence, $\y$ and $\w$ can be connected by a path in $B_1\subset
\overline{B} \cap \Pi_{i-1}^{-1}(\Pi_{i-1}(\y))$. Since $\y$ is in
$\mathscr{C}'$, $\Pi_{i-1}^{-1}(\Pi_{i-1}(\y))$ is contained in
$\mathscr{C}'$ too, and thus in $\mathscr{R}$. Connecting $\y$ to $\w$
and $\w$ to $\z$ (previous paragraph), we conclude the proof of our
claim.

\medskip\noindent{\bf Case 2.} Let now $\y$ be in $\overline{B} \cap
(W_i-\mathscr{C}')$; as in case 1, we assume that $\y$ is not in $B$,
so that~$\Pi_1(\y)=x$.  Since $\y$ is not in $\mathscr{C}'$, $\y$ is
not in $\mathscr{C}$, and so not in $\crit(\Pi_1,W_i)$. Applying
Lemma~\ref{7.3} to the algebraic set $W_i$, we see that $\y$ is in
$\overline{{W_i}_{<x}}$.  By the curve selection lemma, this means
that there exists a semi-algebraic path $\gamma : [0,1] \to {W_i}$
connecting a point $\z$ in ${W_i}_{< x}$ to $\y$, with $\gamma(0)=\z$,
$\gamma(1)=\y$ and $\gamma(t) \in {W_i}_{<x}$ for $t < 1$.

The image of $\gamma$ is in $\mathscr{R}$, so to conclude, it suffices
to prove that $\gamma(t)$ is in $\overline{B}$ for all $t$. To do so,
we will prove that $\gamma(t)$ is in $B$ for all $t<1$. We know that
the image $\{\gamma(t) \ | \ t \in [0,1)\}$ is connected and contained
in $V_{<x}$; hence, it is contained in a connected component $B'$ of
$V_{<x}$. We have to prove that $B'=B$. Because $\gamma(1)=\y$, we
deduce that $\y$ is in $\overline{B'}$; on the other hand, we know
that $\y$ is in $\overline{B}$. Since $\y$ is not in $\mathscr{C}$, it
is not in $W_1$; as a consequence, we can apply Lemma~\ref{lemma:10},
which shows that $B=B'$, as requested.


\subsection*{Proof of Lemma~\ref{lemma:a123} on page~\pageref{lemma:a123}}
Property $\a_1(\e)$ follows from the algebraic form of Sard's lemma;
it is in~\cite{BaGiHeMb97}. Using our notation, Mather's
transversality result~\cite{Mather73,AlOt92,AlBaOt01} shows that for
generic $\e$, $\a_2(\e)$ and $\a_3(\e)$ are satisfied, and the
dimensions of $S_j$ and $S_{j,\ell}$ are
 $$\dim(S_j)=n-1-\nu_{n,i}(n-1-j),\quad  
 \dim(S_{j,\ell})=n-1-\nu_{n,i}(n-1-j,\dim(S_j)-\ell),$$ where the
 function $\nu_{n,i}$ is defined as follows. Considering two indices
 $r \ge s \ge 0$, we define $\mu(r,s)$ as the number of sequences $r'
 \ge s' \ge 0$, with $r' > 1$, and $r \ge r'$, $s \ge s'$; explicitly,
 $\mu(r,s) = r(s+1)-s(s-1)/2$. Then, we have
\begin{eqnarray*}
\nu_{n,i}(r)&=&(i-n+1+r)r\\
\nu_{n,i}(r,s) &=& (i-n+1+r)\mu(r,s) - (r-s)s\\
&=& (i-n+1+r)(r(s+1)-\frac{s(s-1)}2) - (r-s)s.
\end{eqnarray*}
It remains to check that under these constraints, we always have
$\dim(S_{j,\ell}) \le \ell$ for $\ell \le i-1$; this follows from a
straightforward but tedious verification.


\subsection*{Proof of Lemma~\ref{prop:dimxe} on page~\pageref{prop:dimxe}}

In all the rest of this paragraph, we fix $\e$ that satisfies the
assumptions of Lemma~\ref{prop:dimxe}, and we denote by $X_\e$ the
intersection $X \cap \gamma^{-1}(\e)$. Finally, we let $\beta_\e:
(\g,\e,\x) \in X_\e \mapsto \x \in \C^n$ be the projection on the
$\X$-coordinate.
\begin{lemma}\label{lemma:sum}
  For $\x$ in $\reg(W_\e)$ and $\g$ in $\C^i$, $(\g,\x)$ is in $X_\e$
  if and only if $\x$ is in $\crit(\rho_\g \circ \Pi_\e,W_\e)$ and the
  equality $\dim(\Pi_\e(T_\x W_\e))+\dim (\beta_\e^{-1}(\x))=i$ holds.
\end{lemma}
\begin{proof}
  Since ${\bf a}_1(\e)$ holds,, the equations
$f(\X),M_{i+1}(\e,\X),\dots,M_n(\e,\X)$
define the critical set $W_\e$ and for $\x$ in $\reg(W_\e)$, the
matrix $J(\x)$ has rank $n-i+1$. The first claim follows readily.
Thus, $\g$ is in $\beta_\e^{-1}(\x)$ if and only if for all $\v$ in
$T_\x W_\e$, $\rho_\g(\Pi_\e(\v))=0$; equivalently, if for all $\w$ in
$\Pi_\e(T_\x W_\e)$, $\rho_\g(\w)=0$. Since $\rho_\g(\w)=\g \cdot \w$,
we are done.
\end{proof}

For $0 \le \ell \le i-1$, let $j_{\ell,1},\dots,j_{\ell,\kappa(\ell)}$
be the indices $j$ such that $S_{j,\ell}$ is well-defined. Then, we
define the constructible sets
\begin{equation*}
  \label{eq:T}
T_\ell = S_{j_{\ell,1},\ell} \cup \cdots \cup S_{j_{\ell,\kappa(\ell)},\ell}
\quad\text{and}\quad
  T'_\ell = T_0 \cup \cdots \cup T_\ell.
\end{equation*}
By Lemma~\ref{lemma:a123}, both $T_\ell$ and $T'_\ell$ are disjoint
unions of non-singular locally closed sets of dimension at most
$\ell$. By Lemma~\ref{lemma:sum}, for $0 \le \ell \le i$, and for $\x$
in $T_{\ell}$, the inequality $\dim (\beta_\e^{-1}(\x)) \le i-\ell$
holds.  Remark that $W_\e=T'_{i-1} \cup \sing(W_\e)$. Since $T'_{i-1}
= T'_{i-2} \cup T_{i-1}$, we rewrite this as
\begin{equation}
  \label{eq:wi}
W_\e=T'_{i-2} \cup T_{i-1} \cup \sing(W_\e),  
\end{equation}
where the union is disjoint. Going further, we can write for any $\ell \le i-1$
\begin{equation}\label{eq:wi2}
T'_\ell \cup \sing(W_\e) = T'_{\ell-1} \cup T_{\ell} \cup \sing(W_\e).
\end{equation}

Consider now an irreducible component $X'$ of $X_\e$. By construction,
$\beta_\e(X')$ is contained in $W_\e$. By~\eqref{eq:wi}, either
$\beta_\e(X')$ is contained in $T'_{i-2} \cup \sing(W_\e)$, or
$\beta_\e(X')$ intersects $T_{i-1}$. If $\beta_\e(X')$ intersects
$T_{i-1}$, then there is a fiber of dimension at most~1. In this case,
by the theorem on the dimension of fibers, $\dim(X')\le 1+
\dim(T'_{i-1} \cup \sing(W_\e))$, and thus $\dim(X') \le i$.

If $\beta_\e(X')$ is contained in $T'_{i-2} \cup \sing(W_\e)$, then
by~\eqref{eq:wi2}, either $\beta_\e(X')$ is contained in $T'_{i-3}\cup
\sing(W_\e)$, or $\beta_\e(X')$ intersects $T_{i-2}$. If
$\beta_\e(X')$ intersects $T_{i-2}$, then there is a fiber of
dimension at most 2, so $\dim(X')\le 2+ \dim(T'_{i-2} \cup
\sing(W_\e)) \le i$. Continuing this way, we prove that $\dim(X')\le
i$.


\subsection*{Proof of Lemma~\ref{lemma:13} on page~\pageref{lemma:13}}

Let ${\cal F}$ be the Zariski-open subset of $\C^{ni}$ underlying
Lemma~\ref{lemma:a123}: for $\e$ in ${\cal F}$, $\a_1(\e)$,
$\a_2(\e)$ and $\a_3(\e)$ hold. Finally, recall the definitions of the
projections $\alpha: (\g,\e,\x) \mapsto (\g,\e)$ and $\gamma:
(\g,\e,\x) \mapsto~\e$.
First, $Y$ is obviously Zariski-closed. We continue by proving that it
does not cover all of $\C^i\times\C^{ni}$: it is enough to prove it
componentwise. Thus, we partition the set of irreducible components
$X'$ of $X$ into some sets $E_0 \cup E_1 \cup E_2$, where
\begin{itemize}
\item $E_0$ is the set of irreducible components $X'$ of $X$ such
  that $\gamma(X')$ does not intersect ${\cal F}$;
\item $E_1$ is the set of irreducible components $X'$ of $X$ such
  that $\alpha(X')$ intersects ${\cal F}$ and such that 
  $\alpha(X')$ is dense in $\C^i\times\C^{ni}$;
\item $E_2$ is the set of irreducible components $X'$ of $X$ such
  that $\alpha(X')$ intersects ${\cal F}$ and such that
  $\alpha(X')$ is not dense in $\C^i\times\C^{ni}$.
\end{itemize}
We want to prove that for all $X'$, the set of infinite fibers of
$\alpha$ in $X'$ is contained in a strict Zariski-closed subset of
$\C^i\times \C^{ni}$. For $X'$ in $E_0$, $\gamma(X')$ is contained in
a strict Zariski-closed subset of $\C^{ni}$, which implies that
$\alpha(X')$ is contained in strict Zariski-closed subset of
$\C^i\times \C^{ni}$. For $X'$ in $E_1$, Lemma~\ref{prop:dimxe} and
the theorem on the dimension of fibers imply that $\dim(X') \le i+ni$;
as a consequence, the set of infinite fibers is contained in a
hypersurface. For $X'$ in $E_2$, this is true by construction. This
finishes the proof that $Y$ is a strict Zariski-closed subset of
$\C^i\times\C^{ni}$.

Let $(\g,\e)$ be in $\C^i\times\C^{ni}-Y'$. Hence, $\e$ is in ${\cal
  F}$, so that the fiber $\alpha^{-1}(\g,\e)$ meets no irreducible
component $X'$ of $X$ that belongs to $E_0$.  For all other components
$X'$ of $X$, since $(\g,\e)$ is not in $Y$, $\alpha^{-1}(\g,\e)$
intersects $X'$ in a finite number of points.  Hence, finally,
$\alpha^{-1}(\g,\e)$ intersects $X$ in a finite number of points. By
Lemma~\ref{lemma:sum}, this means that $\crit(\rho_\g\circ
\Pi_\e,W_\e)$ is finite, as requested.


\subsection*{Proof of Lemma~\ref{complexiteCanny} on page~\pageref{complexiteCanny}: correctness and runtime}

In this paragraph, we prove that assuming $\AS$ and $\BS$, algorithm
{\sf CannyRoadmap} is correct; we also discuss its complexity. In all that
follows, for $i\le j$, we denote by $\Pi_{X_i,\dots,X_j}$ the projection 
$$\begin{array}{cccc}
\Pi_{X_i,\dots,X_j}: & \C^n & \to & \C^{j-i+1} \\
 & \x =(x_1,\dots,x_n) & \mapsto & (x_i,\dots,x_j).
\end{array}$$

First, we need a direct extension of Theorem~\ref{theo:big} to the
case of inputs of the form $[\f,Q]$, with $\f=f_1,\dots,f_p$ and
$Q(X_1,\dots,X_e)$, so that we have $d=n-p-e$.  As before, we are also
given a set of control points $\mathscr{P}$ in $V=V(\f,Q)$. Then,
extending the previous notation, we define, for $\x=(x_1,\dots,x_e)$
in $V(Q)$:
\begin{itemize}
\item $V_\x=V(\f(x_1,\dots,x_e,X_{e+1},\dots,X_n))\subset \C^n$;
\item $\mathscr{P}_\x = \mathscr{P} \cap V_\x$;
\item $\mathscr{C}_\x = \crit(\Pi_{X_{e+1}},V_\x)\,\cup \, 
                        \crit(\Pi_{X_{e+1}}, \crit(\Pi_{X_{e+1},\dots,X_{e+i}},V_\x))\, \cup\,
  \mathscr{P}_\x$;
\item $\mathscr{C}_\x'=V_\x\,\cap\,\Pi_{X_{e+1},\dots,X_{e+i-1}}^{-1}(\Pi_{X_{e+1},\dots,X_{e+i-1}}(\mathscr{C}_\x)) =V_\x\,\cap\,\Pi_{X_{1},\dots,X_{e+i-1}}^{-1}(\Pi_{X_{1},\dots,X_{e+i-1}}(\mathscr{C}_\x))$.
\end{itemize}
\noindent If $[\f,Q]$ satisfies $\AS$ and $\BS$, then for all $\x \in V(Q)$,
$\mathscr{C}_\x$ is finite.

\begin{theorem}\label{theo:big2}
  Let $d' = \max(i-1, d-i+1)$. If $[\f,Q]$ satisfies $\AS$ and
  $\BS$, then for all $\x=(x_1,\dots,x_e)$ in $V(Q)$, the following holds:
  \begin{enumerate}
  \item $\mathscr{C}'_\x\cup \crit(\Pi_{X_{e+1},\dots,X_{e+i}},V_\x))$ is a $d'$-roadmap of
    $(V_\x,\mathscr{P}_\x)$;
  \item $\mathscr{C}'_\x\cap \crit(\Pi_{X_{e+1},\dots,X_{e+i}},V_\x))$ is finite;
  \item for all $(x_{e+1},\dots,x_{e+i-1})\in \C^{i-1}$, the system
    $(f_1,\dots,f_p,X_1-x_1,\dots,X_{e+i-1}-x_{e+i-1})$ satisfies
    assumption $\AS$.
  \end{enumerate}
\end{theorem}
\noindent This theorem is a straightforward consequence of
Theorem~\ref{theo:big}, applied to all algebraic sets $V_\x$. With
this in mind, we start by analyzing a single level of algorithm {\sf
  CannyRoadmap}.
\begin{lemma}\label{lemma:5}
  Suppose that $[\f,Q]$ satisfies $\AS$, and that after the change of
  variables $\varphi$, $[\f,Q]$ satisfies $\BS$ for $i=2$. Then steps
  $0-6$ of algorithm {\sf CannyRoadmap} take time $(\delta_{Q} +
  \delta_{P})^{O(1)}(nD)^{O(n)}$; upon success, $Q'$ and $P'$ are
  0-dimensional parametrizations that satisfy
  $$\delta_{Q'} + \delta_{P'} \le  (\delta_{Q} + \delta_{P}) (nD)^{O(n)}$$
  and $[\f,Q']$ satisfies $\AS$. Let finally $\mathscr{P}' \subset
  \C^n$ be the set described by $P'$.  If the recursive call at step 7
  computes a roadmap of $(V(\f,Q'),\mathscr{P}')$, then $(R',R'')$ is
  a roadmap of $(V(\f, Q),\mathscr{P})$.
\end{lemma}
\begin{proof}
  Let us write here $V=V(\f,Q)$. We start by proving correctness.
  Remark that the solution set of $(\f,\Delta,Q)$ is the union of the
  sets $\crit(\Pi_{X_{e+1}},V_\x)$. Similarly, the solution-set of
  $(\f,\Delta',Q)$ is the union of the critical set
  $\crit(\Pi_{X_{e+1},X_{e+2}},V_\x)$, for the projection on the
  $(X_{e+1},X_{e+2})$-axis. Because $\BS$ holds for $i=2$, this set
  has dimension 1. Consequently, $S$ describes the union of the
  critical points of $\Pi_{X_{e+1}}$ on the sets
  $\crit(\Pi_{X_{e+1},X_{e+2}},V_\x)$. Because $\BS$ holds for $i=2$,
  this set is finite.  Then, $Q'$ describes all the projections
  $\Pi_{X_1,\dots,X_{e+1}}(\mathscr{C}_\x)$.

  By the first point of Theorem~\ref{theo:big2}, each $\mathscr{C}'_\x
  \cup \crit(\Pi_{X_{e+1},X_{e+2}},V_\x))$ is a $(n-p-e-1)$-roadmap
  of $(V_\x,\mathscr{P}_\x)$. Besides, $P'$ describes a set
  $\mathscr{P}'$ which is the union of all set $(\mathscr{C}'_\x \cap
  \crit(\Pi_{X_{e+1},X_{e+2}},V_\x)) \cup \mathscr{P}_\x$; it is
  finite by point 2 in Theorem~\ref{theo:big2}.

  We continue by remarking that point 3 in Theorem~\ref{theo:big2}
  shows that $[\f,Q']$ satisfies $\AS$, which justifies the recursive
  call on step 7. Suppose now that we obtain as output a roadmap $R''$
  of $(V(\f,Q'),\mathscr{P}')$, and we write $R''$ as the disjoint
  union of the sets $R''_\x$, for $\x \in V(Q)$, with $R''_\x=R'' \cap
  \Pi_{X_1,\dots,X_e}^{-1}(\x)$. By the claims of the first paragraph,
  the zero-set of $(\f,Q')$ is the union of all $\mathscr{C}'_\x$,
  which implies that each $R''_\x$ is a roadmap of $(\mathscr{C}'_\x,
  (\mathscr{C}'_\x \cap \crit(\Pi_{X_{e+1},X_{e+2}},V_\x)) \cup
  \mathscr{P}_\x)$. Applying Lemma~\ref{lemma:glue}, we deduce that
  each union $R''_\x \cup \crit(\Pi_{X_{e+1},X_{e+2}},V_\x))$ is a
  roadmap of $(V_\x, \mathscr{P}_\x)$. This proves that $(R',R'')$ is
  a roadmap of $(V,\mathscr{P})$.

\medskip

Next, we estimate the degree of the output, assuming
correctness. First, we fix $\x$ in $V(Q)$ and bound the degree of the
various objects above $\x$, leaving aside the contribution of
$\mathscr{P}$ for the moment. By B\'ezout's theorem, $V_\x$ has degree
at most $D^p$, whereas the degrees of $\crit(\Pi_{X_{e+1}},V_\x)$ and
$\crit(\Pi_{X_{e+1},X_{e+2}},V_\x)$ are at most $D^p
(pD)^{n-p}=p^{n-p}D^n$ (the latter estimate relies on the B\'ezout
theorem of~\cite[Prop.~2.3]{HeSc80}).  Finally, since
$\crit(\Pi_{X_{e+1},X_{e+2}},V_\x)$ is a curve of degree at most
$p^{n-p}D^n$, the set of critical points of $\Pi_{X_{e+1}}$ on this
curve has degree at most $p^{2n-2p} D^{2n}$.

Taking all $\x$ in $V(Q)$ into account, we deduce that the degrees of
$R$ and $R'$ are both bounded by $\delta_Q p^{n-p} D^n$ and the degree
of $S$ is at most $\delta_Q p^{2n-2p} D^{2n}$, so that the degree of
$Q'$ is at most $2\delta_Q p^{2n-2p} D^{2n}+\delta_{P}$.

It remains to bound the degree of $P'$; we start by estimating the
degree of ${\sf Solve}([Q',R'])$, which computes the intersection of
$\Pi_{X_1,\dots,X_{e+1}}^{-1}(\Pi_{X_1,\dots,X_{e+1}}(V(S)\cup
V(R)\cup \mathscr{P}))$ with the zero-set of $R'$. Above each value of
$\x$ in $V(Q)$, the intersection has degree at most $(\delta_{P_\x} +
2p^{2n-2p} D^{2n}) p^{n-p}D^n$, where $\delta_{P_\x}$ is the
cardinality of $\mathscr{P}_\x$. Summing over all $\x$ in $V(Q)$ gives
the upper bound $$\delta_P p^{n-p}D^n + \delta_Q 2p^{3n-3p} D^{3n}$$
for the degree of ${\sf Solve}([Q',R'])$, and thus $$\delta_P
(1+p^{n-p}D^n) + \delta_Q 2p^{3n-3p} D^{3n}$$ for the degree of $P'$.
Taking into account the estimate on the degree of $Q$, we obtain the
upper bounded announced in the lemma.

\medskip

Finally, we estimate the running time, starting with the computation
of $R$ and $R'$. If we were to solve a system of the form
$[\f,\Delta,X_1-x_1,\dots,X_e-x_e]$, the resolution algorithm
of~\cite{Lecerf00} would take time $(nD)^{O(n)}$. However, we need to
solve slightly more complex systems of the form $[\f,\Delta,Q]$ or
$[\f,\Delta',Q]$. Our strategy is to use dynamic evaluation
techniques~\cite{D5}: we apply the former algorithm over the product
of fields $\Q[T]/q$, where $q$ is the minimal polynomial of $Q$. If a
division by zero occurs, we split $q$ into two factors, and we run the
computation again. The maximal number of splittings is $\delta_Q$, so
the overall cost is $\delta_Q^{O(1)}(nD)^{O(n)}$.

The critical points computation takes a similar time, since the form
of the parametrization makes it possible for us to work with bivariate
polynomials of degree $(nD)^{O(n)}$. The union and projection at step
5 take time $(\delta_Q+\delta_{P})^{O(1)}(nD)^{O(n)}$, since they
only involve computations with 0-dimensional ideals of that degree,
given by rational parametrizations, and rational parametrizations for
such objects can be computed (deterministically) in the required time
using e.g. the algorithm of~\cite{Rouillier99}. 
Solving the system $[\f,\Delta,Q']$ is done by the same dynamic
evaluation strategy as before, and the final union computation 
raises no new difficulty.
\end{proof}

\medskip\noindent Remark that as soon as all changes of variables
satisfy the assumptions of Lemma~\ref{prop:3}, the previous lemma
shows that the whole algorithm {\sf CannyRoadmap} correctly computes a
roadmap of $V([\f,Q],\mathscr{P})$ in the requested time (the analysis
of the overall computation time in on page~\pageref{complexiteCanny}). The
probabilistic aspects are discussed further.


\subsection*{Proof of Lemma~\ref{complexite} on page~\pageref{complexite}: correctness and runtime}

The proof of the running time estimates for our algorithm is quite
similar to that given for our modified version of Canny's algorithm.
In what follows, to simplify notation, we denote by $\f$ the system
$[f,\Delta']=[f,\partial f/\partial X_{e+i+1},\dots,\partial
f/\partial X_n]$ used in the algorithm.

\begin{lemma}\label{lemma:ourC}
  Suppose that $[f,Q]$ satisfies $\AS$ and that after the change of
  variables $\varphi$, $[f,Q]$ satisfies~$\BS$ and $[\F,Q]$ satisfies
  $\AS$. Then steps $0-7$ of algorithm {\sf Roadmap} take time
  $(\delta_{Q} + \delta_{P})^{O(1)}(nD)^{O(n)}$; upon success, $Q'$
  and $P'$ are 0-dimensional parametrizations that satisfy
  $$\delta_{Q'} + \delta_{P'} \le (nD)^{O(n)} (\delta_{Q} + \delta_{P})$$
  and $[f,Q']$ satisfies assumption $\AS$.  Let finally $\mathscr{P}'
  \subset \C^n$ be the set described by $P'$. If additionally
  \begin{itemize}
  \item the call to {\sf CannyRoadmap} at step 8 computes a roadmap $R''$ of $(V(\f,Q),\mathscr{P}')$,
  \item the recursive call at step 9 computes a roadmap $R'''$ of $(V(f,Q'),\mathscr{P}')$,
  \end{itemize}
  then $(R'',R''')$ is a roadmap of $(V(f, Q),\mathscr{P})$.
\end{lemma}
\begin{proof}
  The proof follows exactly the same pattern as the one in
  Lemma~\ref{lemma:5}. The only notable difference is that we directly
  use the defining system $[\F,Q]$ to compute the critical points of
  $\Pi_{X_{e+1}}$, which is possible since these equations satisfy $\AS$.
\end{proof}

\medskip\noindent As for algorithm {\sf CannyRoadmap}, as soon as all
changes of variables satisfy the assumptions of Lemma~\ref{prop:4},
the previous lemma shows that the whole algorithm {\sf Roadmap}
correctly computes a roadmap of $V([f,Q],\mathscr{P})$ in the
announced time.


\subsection*{Probabilistic aspects of our algorithms}\label{sec:probasp}

Both algorithms {\sf CannyRoadmap} and {\sf Roadmap} start by choosing
a random change of variable $\varphi$ in a parameter space denoted by
$\GL(n,e)$. Lemmas~\ref{prop:3} and~\ref{prop:4} show that success
depends on choosing $\varphi$ outside of some hypersurfaces of
$\GL(n,e)$; what is missing is an estimate on the degrees of these
hypersurfaces.

Let us assume that we initially call {\sf Roadmap} with input a
polynomial $f$ of degree $D$, $Q$ of degree $\delta_Q$ and $P$ of
degree $\delta_P$; the following lemma gives a bound on the degree of
the hypersurface to avoid which is valid at any step of the
recursion. We give the bound in a big-O form for readability; all
estimates could be made completely explicit.
\begin{lemma}\label{lemma:Q}
  Starting with conditions as above, at any recursive call to {\sf
    CannyRoadmap} (resp. {\sf Roadmap}), there exists a hypersurface
  $H$ of degree at most $K(n,D,\delta_P,\delta_Q)=(\delta_Q+\delta_P)
  D^{O(n^2)}$ of $\GL(n,e)$ such that if $\varphi \in H$, the
  conclusions of Lemma~\ref{prop:3} (resp. Lemma~\ref{prop:4}) are
  satisfied.
\end{lemma}
\begin{proof}
  A useful ingredient is a quantitative version of Sard's
  lemma~\cite[Prop.~3.6]{Mumford76}.
\begin{lemma}\label{lemma:sardQ}
  Suppose that $X \subset \C^K$ is an algebraic set defined by
  equations of degree $\Delta$ and that $\Phi: X \to\C^L$ is a
  polynomial map, given by means of equations of degree $\Delta$ as
  well.  Suppose that $K,L \le N$; then, $\Phi(\reg(X) \cap
  \crit(\Phi,X))$ is contained in a hypersurface of $\C^L$ of degree
  $\Delta^{O(N^2)}$.
\end{lemma}
\begin{proof}
  Remark that $X$ has degree at most $\Delta^N$. First we show that we
  can write $\reg(X) \cap \crit(\Phi,X)$ as $\reg(X)\cap Z$, for a
  suitable algebraic set $Z \subset X$. Let $X'$ be the reunion of the
  irreducible components of $X$ of maximal dimension $d$; we know that
  $X'$ can be generated by $O(N)$ polynomials $g_1,\dots,g_R$ of
  degree at most $\Delta^N$, by~\cite[Prop.~3]{Heintz83}. Then, we
  define $Z$ by $g_1,\dots,g_R$ and all $(K+L-d)$-minors of the
  Jacobian matrix matrix of $(g_1,\dots,g_R,\Phi_1,\dots,\Phi_L)$, and
  we easily verify the claim that $\reg(X)
  \cap\crit(\Phi,X)=\reg(X)\cap Z=(X-\sing(X))\cap Z$.

  By B\'ezout's theorem as in~\cite[Prop.~2.3]{HeSc80}, we obtain the
  bound $\Delta^{O(N^2)}$ for the degree of $Z$. Now, since $Z$ is
  contained in $X$, we can rewrite $\reg(X)\cap \crit(\Phi,X)$ as $Z -
  \sing(X)\cap Z$. Consequently, a degree bound as above hold for the
  degree of the Zariski closure of $\reg(X) \cap \crit(\Phi,X)$, and
  for the degree of its image by~$\varphi$ (by B\'ezout's theorem
  again).
\end{proof}

We can now resume the proof of Lemma~\ref{lemma:Q}.  Each time we
enter the functions {\sf CannyRoadmap} and {\sf Roadmap}, the input
polynomials (either the system $\f=f_1,\dots,f_p$ or the unique
equation $f$) have degree at most $D$ and the 0-dimensional
parametrization $Q$ has degree $(\delta_Q+\delta_P)
(nD)^{O(n^{1.5})}$. We consider all $\x$ in $V(Q)$ separately: each of
them puts some constraints on $\varphi$, and $\varphi$ must satisfy
all of these constraints simultaneously.

If we prove that for a single $\x \in V(Q)$ the degree of the
hypersurface to avoid in $\GL(n,e)$ is $D^{O(n^2)}$, then the degree
of the union of all these hypersurfaces will be $(\delta_Q+\delta_P)
D^{O(n^2)}$, as claimed. Concretely, after fixing
$\x=(x_1,\dots,x_e)$, we are left to quantify the claims that proved
Lemmas~\ref{prop:3} and~\ref{prop:4} in Section~\ref{sec:gen}; we
apply them to the variety $V_\x$ defined by the input polynomials and
the additional equations $X_1=x_1,\dots,X_e=x_e$. Note that all these
equations have degree at most $D$.

The first step is a dimension statement for polar varieties in generic
coordinates. This is proved in~\cite{BaGiHeMb97,BaGiHeMb01} by means
of an algebraic version of Thom's weak transversality result, applied
to a generic projection $\Phi$ on $V_\x$. The weak transversality
theorem is obtained by applying Sard's lemma to a subset $S$ of $V_\x
\times Y$, where $Y$ is the parameter space where we pick our generic
projection and $S$ is defined by equations of degree $O(D)$. By
Lemma~\ref{lemma:sardQ}, we obtain the degree bound $D^{O(n^2)}$ for
the critical locus, as claimed.

The second step is a Noether position statement for polar varieties.
Using a change of variables with formal entries (that is, new
variables ${\bf U}$), we construct in~\cite[Sect.~2.3]{SaSc03} some
eliminating polynomials with coefficients that are rational functions
of ${\bf U}$. Besides, we prove in~\cite[Sect.~2.4]{SaSc03} that if
the entries of the change of variables $\varphi$ cancel none of the
denominators of these coefficients, the polar varieties associated to
$V_\x$ are in Noether position. The least common multiple of these
denominators has degree $D^{O(n)}$ by~\cite[Prop.~1]{Schost03}; this
gives the degree bound for this step as well.

As seen in Section~\ref{sec:gen}, this is sufficient to conclude for
Lemma~\ref{prop:3}. The most delicate step is to establish point $(d)$
of assumption $\BS$ for Lemma~\ref{prop:4}. Recall that in
Section~\ref{sec:gen} we defined a strict algebraic $Y$ of $\C^i \times
\C^{ni}$, such that the first $i$ rows of the inverse of $\varphi$
should avoid $((1,0,\dots,0)\times \C^{ni}) \cap Y$. Hence, it is
sufficient to bound the degree of $Y$ by $D^{O(n^2)}$.

We reconsider the proof given above of Lemma~\ref{lemma:13} (and use
freely all necessary notation). First, observe that the algebraic set
$X$ defined on page~\pageref{lemma:a123} has degree $D^{O(n)}$. We
also recall that $Y$ consists of the Zariski-closure of the infinite
fibers of a projection denoted by $\alpha: X \to \C^i\times \C^{ni}$.
The irreducible components of $X$ were classified into three groups,
written $E_0$, $E_1$ and $E_2$. We prove that in all cases, the
Zariski-closure of the set of the infinite fibers of $\alpha$ on $X'$
has a degree at most that of $X'$.
\begin{itemize}
\item The image of a component $X'$ in $E_0$ is contained in a strict
  algebraic subset of $\C^i \times \C^{ni}$; then, it can be enclosed
  in a hypersurface of degree bounded by that of $X'$. The same holds
  for the components in $E_2$.
\item For a component $X'$ in $E_1$, we saw that the projection
  $\alpha:X' \to \C^i\times \C^{ni}$ has a dense image and generically
  finite fibers. Let $\C(V_1,\dots,V_{n+ni})$ be the function field of
  $\C^i\times \C^{ni}$, let $\C(X')$ be that of $X'$, and let $M \in
  \C(V_1,\dots,V_{i+ni})[T]$ be the monic minimal polynomial of a
  primitive element for the algebraic extension
  $\C(V_1,\dots,V_{i+ni}) \to \C(X')$. It is known that the infinite
  fibers cancel one of the denominators of the coefficients of
  $M$~\cite{SaSc03}. Since the least common multiple of these
  denominators has degree at most the degree of $X'$~\cite{Schost03},
  we are done.
\end{itemize}
At this stage, we have quantified Lemma~\ref{prop:3} and the first
part of Lemma~\ref{prop:4}; it remains to consider the last condition
of that lemma (that the system $[\F,Q]$ satisfies assumption $\AS$).
We mentioned in Section~\ref{sec:gen} that this property resulted from
the validity of a condition written $\a_1(\e)$, which itself is
ensured by an application of Sard's lemma. The quantification is
similar to the one we have seen before, and yields another
contribution of the form $D^{O(n^2)}$.
\end{proof}

We conclude the probability analysis of our algorithms. At each level
of the recursion, we draw all entries of our change of variables in a
set of cardinality $\eta\, K(n,D,\delta_Q,\delta_P)$, where
$K(n,D,\delta_Q,\delta_P)$ was defined in the previous lemma. By
Zippel-Schwartz's zero avoidance lemma, the probability of success at
this level is at least $(1-1/\eta)$. We need to draw at most $n^2$
changes of variables; hence, to obtain an overall probability of
success of at least $1/2$, it suffices to take $\eta$ polynomial in
$n$.

It remains to discuss the probabilistic aspects of the algorithm
of~\cite{Lecerf00}; they are twofold. First, the success of that
algorithm depends on the choice of a so-called {\em correct test
  sequence}~\cite{HeSc80}, to perform zero-test of polynomials represented by
straight-line programs. For all our applications of this subroutine, a
single correct test sequence is needed; as pointed out
in~\cite{FiGiSm95}, one can construct one with probability of success
at least $1/262144$. The second probabilistic aspect is due to a
linear combination of the equations performed at the beginning of this
subroutine. This aspect is analyzed in~\cite{Lecerf00}. The conclusion
is similar to what we obtained above for our change of variables:
success is ensured if the coefficients of the linear combination avoid
a hypersurface of degree $D^{O(n)}$.

\end{document}